%% file: arxiv-210821.tex
\documentclass[aip,jmp,floats,twocolumn,reprint]{revtex4-1}
\usepackage{amsmath,amsfonts,amssymb,amsthm,stmaryrd,bm}
\usepackage{mathtools,wasysym,mathrsfs}
\usepackage{accents}
\usepackage{graphicx}
\usepackage{hyperref}
\usepackage{breakurl}
\usepackage{color}
\usepackage{tikz-cd}
\usepackage{epigraph}
\usepackage{todonotes}

\newcounter{dummy}
\usepackage{enumitem}
\makeatletter
\newcommand\myitem[1][]{\item[#1]\refstepcounter{dummy}\def\@currentlabel{#1}}
\makeatother

\theoremstyle{plain}
\newtheorem{thm}{Theorem}[section]
\newtheorem{cor}[thm]{Corollary}
\newtheorem{lem}[thm]{Lemma}
\newtheorem*{lem*}{Lemma}
\newtheorem*{cor*}{Corollary}
\newtheorem*{thm*}{Theorem}
\newtheorem{prop}[thm]{Proposition}

  \theoremstyle{definition}
\newtheorem{defn}[thm]{Definition}
\newtheorem{notn}[thm]{Notation}
\newtheorem{cnvntn}[thm]{Convention}

\theoremstyle{remark}
\newtheorem{rem}[thm]{Remark}

\let\origmaketitle\maketitle
\def\maketitle{
  \begingroup
  \def\uppercasenonmath##1{} 
  \let\MakeUppercase\relax 
  \origmaketitle
  \endgroup
}

\input{macros}

\begin{document}

\title{Infinite-dimensional analyticity in quantum physics}

\author{Paul~E.~Lammert}
\email{lammert@psu.edu}
\affiliation{Department of Physics, 104B Davey Lab \\ 
Pennsylvania State University \\ University Park, PA 16802-6300}
\begin{abstract}
  A study is made, of families of Hamiltonians parameterized over open subsets of
  Banach spaces in a way which renders many interesting properties of
  eigenstates and thermal states
  analytic functions of the parameter. Examples of such properties are charge/current
  densities. The apparatus can be considered a generalization of Kato's theory of
  analytic families of type B insofar as the parameterizing spaces are infinite dimensional.
  It is based on the general theory of holomorphy in Banach spaces and an identification
  of suitable classes of sesquilinear forms with operator spaces associated with
  Hilbert riggings. The conditions of lower-boundedness and reality appropriate to
  proper Hamiltonians is thus relaxed to sectoriality, so that holomorphy can be used.
  Convenient criteria are given to show that a parameterization
$x \mapsto {\mathsf{h}}_x$ of sesquilinear forms is of the required
sort ({\it regular sectorial families}).
The key maps \hbox{${\mathcal R}(\zeta,x) = (\zeta - H_x)^{-1}$} and
${\mathcal E}(\beta,x) = e^{-\beta H_x}$, where $H_x$ is the closed sectorial operator
associated to $\frm{h}_x$, are shown to be analytic. These mediate analyticity of
the variety of state properties mentioned above.
A detailed study is made of nonrelativistic quantum mechanical Hamiltonians
parameterized by scalar- and vector-potential fields and two-body interactions.
\end{abstract}
\date{Aug. 22, 2021}

\maketitle
\tableofcontents

\newpage
\section{Introduction}

\subsection{Motivation}

The mathematical concept of analyticity is ubiquitous in physics. Here is a short list of examples.
It is in the background whenever we approximate a function by a few terms of its Taylor series.
The question of whether perturbation series converge or not is of interest in many contexts.
Kramers-Kr\"{o}nig relations are a manifestation of analyticity in
complex half-planes. In thermodynamics, phase transitions are identified
with the locus of points in a phase diagram at which free energy fails to
be analytic. In quantum mechanics, analyticity of a resolvent operator in the
spectral parameter is important. 
Those examples, and most other applications, consider regularity with respect to a few variables,
often just one.
This paper is concerned with analyticity when both domain and codomain are infinite-dimensional.
Functional Taylor expansions, so-called, are used in the physics literature,
but in a purely formal way so that one is hard-pressed to say anything about their
existence or what convergence would even amount to.

The original motivation for this investigation emerged from density functional
theory\cite{Koch+Holthausen,Capelle06,Dreizler+Gross,Parr+Yang,Burke-12} (DFT),
which is the foundation for very practical and successful computation in
solid-state physics, chemistry and materials science.
The connection to the present work is briefly described to illustrate ``real-world'' relevance.
One considers the ground-state energy $E(v)$ of an $N$-electron system as a function of an
``arbitrary'' external one-body potential $v$.
Alternatively (and preferentially in DFT) one focuses on the intrinsic energy $F(\rho)$
which is the minimum kinetic-plus-Coulomb-interaction energy consistent with charge
density $\rho$.
$E(v)$ and $F(\rho)$ stand in a relation of Legendre duality to each other, and
their arguments range over certain infinite-dimensional spaces.
$F(\rho)$ is {\em everywhere discontinuous}, and one may not expect much better
of $E(v)$ due to the duality relation.
Surprisingly, that is very far from true.
For instance (see Section~\ref{sec:eigenstate-cc-density}),
for an energetically isolated nondegenerate eigenstate (not only ground states),
charge density is {\em analytic} in $L^3(\Real^3)\cap L^1(\Real^3)$
as function of scalar potential $v$ in $L^{3/2}(\Real^3)\cap L^\infty(\Real^3)$.
Thus, as a function of $v$, $\rho$ is so smooth that it has a convergent Taylor series.
This fact has significant implications for computational practice,
which, however, are outside the scope of this paper and will be taken up elsewhere.
Other results have implications for less-common flavors of DFT such as
current-density functional theory\cite{Vignale+Rasolt-88,Grayce+Harris-94}
and non-zero-temperature DFT\cite{Mermin-65,Dornheim+18}.

More generally, 
suppose we have a family of quantum Hamiltonians parameterized in a natural way by
parameter $x$ ranging over an open subset of a Banach space.
Under what conditions are physically interesting quantities analytic functions of $x$?
Such quantities pertaining to an eigenstate include: the state itself,
the corresponding energy eigenvalue, expectations of observables and generalized
observables such as charge/current density.
And, for nonzero temperature: 
statistical operator (i.e., the thermal state), free energy, thermal expectations,
susceptibilities, and so on.
The framework developed here can be used to address such analyticity questins with
relative ease, as is demonstrated explicitly.
The framework is flexible, powerful, and general due to treating Hamiltonians initially as
sesquilinear forms, with a relaxation of the physically-grounded requirements of reality and
lower-boundedness to sectoriality so that holomorphy ($\Cmplx$-differentiability) 
can be invoked, and complex analysis methods brought to bear. 

Kato's analytic perturbation theory for type B families\cite{Kato}
is concerned with similar questions, but only for families with parameterization
domains in the complex numbers $\Cmplx$.
The move to infinite-dimensional parameterization domains (Banach spaces, specifically)
not only increases flexibility, but triggers a conceptual rearrangement, leading to
a rephrasing of everything in terms of compositions of holomorphic maps between
Banach spaces.
It therefore becomes imperative to repackage appropriate classes of unbounded sesquilinear
forms as Banach spaces. Section~\ref{sec:families} develops that key part of the apparatus.

\subsection{An operator prototype}

It is a familiar and useful fact that the resolvent
$\Rmap(\zeta,H) = (H-\zeta)^{-1}$ of operator $H$
is a holomorphic function of the spectral parameter $\zeta$.
The extension to a holomorphic dependence on $H$ is worth looking at as a prototype for
the theory to be developed.
\begin{defn}
  \label{def:relative-operator-bound}
  Given: a closed, densely-defined, operator $T$ on Banach space $\sX$
  [denoted $T\in\Lincl(\sX)$].
An operator $A$ is {\it $T$-bounded}
if $\dom A \supseteq \dom T$ and there are $a,b$ such that
\begin{equation}
\forall x \in \dom T, \; \|A x\| \le a \|x\| + b\|Tx\|.
\end{equation}
By increasing $a$, it may be possible to decrease $b$. The infimum of all
$b$'s that work is the {\it $T$-bound} of $A$.
\end{defn}

If $T$ is closed and invertible, $\dom T$ can be turned into a Banach space
with the norm $\|x\|_T = \|Tx\|_\sX$; we understand it as such in the following.
The following Lemma brings the notion of analyticity to the surface.
(A proof is provided at the end of the subsection, which should probably be
skipped until it is called upon.)
\begin{lem}
\label{lem:inverse-of-perturbed-closed-op}
Suppose $T$ is closed with range {$\sX$}, and $A$ is $T$-bounded.
If \hbox{$\rng(T+A) = \sX$}, then
  \begin{equation}
\label{eq:perturbed-inverse}
(T+A)^{-1}= T^{-1}(1 + AT^{-1})^{-1} \in \Lin(\sX).
  \end{equation}
This holds in particular when $\|AT^{-1}\| < 1$, which 
implies convergence of the Neumann series
\begin{equation}
  \label{eq:Neumann-series}
(T+A)^{-1}= T^{-1}\sum_{n=0}^\infty (-AT^{-1})^n.
\end{equation}
\end{lem}
One would like to hold up series (\ref{eq:Neumann-series}) as the
demonstration that $(T+A)^{-1}$ is analytic at $A=0$. That the terms are
not simply multiples of powers of $A$, however, shows the need for at least the
rudiments of a more general theory of analyticity, a theory which will be reviewed
in Section \ref{sec:Banach-holomorphy}. Similarly, a claim that the series converges
uniformly on some ball about the origin raises the question of domain and codomain
of the map. The codomain is clearly $\Lin(\sH)$. The domain {\em could} be taken
the same, but that would be far too timid. Instead, consider $\dom T$ as a Banach
space with norm $\|x\|_T = \|Tx\|$, making $T$ an isomorphism. Then,
$A\mapsto (T+A)^{-1}$ can be considered a map from $\Lin(\dom T;\sX)$ to $\Lin(\sX)$.
Indeed, the norm of $A$ as an element of $\Lin(\dom T;\sX)$ is precisely
$\|AT^{-1}\|$, so the series (\ref{eq:Neumann-series}) is uniformly convergent
on any center-zero ball of radius less than 1.

It might seem very difficult to extend this to families of operators having
{\em differing} domains, but we shall find it possible by working with
Hamiltonians in the guise not of operators, but of {\it sesquilinear forms}.
Recall that a sesquilinear form (\sqf) in Hilbert space $\sH$ is
a complex-valued function ${\frm{h}}[\phi,\psi]$, linear in $\psi$
and conjugate-linear in $\phi$ which range over some subspace of $\sH$.
Motivations for working with sesquilinear forms are, first, the increased strength.
That is needed for DFT applications, for example, where potentials in $L^{3/2}(\Real^3)$
are considered. Secondly, there is a corresponding gain in flexibility in applications, as
it becomes easier to verify that a family of {\sqf}s is appropriate for
feeding into the automatic abstract machinery. This is illustrated in Section \ref{sec:QM}.
And finally, there is the argument that {\sqf}s are more physically natural and
meaningful than operators. For, an \sqf\ can be recovered from its
diagonal elements, and as expectation values these have a far clearer operational
meaning than multiplication by a Hamiltonian operator.

\begin{proof}[Proof of Lemma~\ref{lem:inverse-of-perturbed-closed-op}]
$T+A$ is closed on $\dom (T+A) = \dom T$
by Lemma~\ref{lem:stability-of-closedness} below,
and $AT^{-1} \in \Lin(\sX)$ since
$\|AT^{-1} x\| \le \Big( a\|T^{-1}\| + b \Big) \|x\|$.

Now, \hbox{$(T+A) =  (T+A)T^{-1}T = (1+AT^{-1})T$}
gives \hbox{$\rng(T+A) \subseteq \rng(1+AT^{-1})$}.
The reverse inclusion follows from
\hbox{$(T+A)T^{-1} = (1+AT^{-1})$}.
Thus, if either $T+A$ or $1+AT^{-1}$ has a (necessarily bounded) inverse,
so does the other, and (\ref{eq:perturbed-inverse}) holds.
\end{proof}
\begin{lem}
  \label{lem:stability-of-closedness}
If $A$ has $T$-bound strictly less than one,
Then $T+A$ is closed on $\dom T$.
\end{lem}
\begin{proof}
For any $\psi\in\dom T$,
\begin{equation}
\label{eq:eq-lem-stability}
| \|T\psi\| - \|(T+A)\psi\| | \le \|A\psi\| \le a\| \psi \| + b \|T\psi\|,
\end{equation}
which yields
$(1-b) \|T\psi\| \le a\|\psi\| + \|(T+A)\psi\|$,
after rearrangement. 
Suppose that sequence $(\psi_n)$ in $\dom T$ converges to zero and
\hbox{$((T+A)\psi_n)$} is Cauchy.
$(T\psi_n)$ is also Cauchy by the preceding inequality, with
(because $T$ is closed) limit zero.
But, then (\ref{eq:eq-lem-stability}) shows that 
\hbox{$(T+A)\psi_n \to 0$}, as well.
\end{proof}

\subsection{Sketch of the theory}

The main ideas are sketched in this subsection, made more concrete with the
aid of the example of a nonrelativistic ``spinless electron''
subjected to a variable external vector potential field ${\bm A}(x)$ ($x\in\Real^3$).
This application will be treated in depth in Section~\ref{sec:QM};
here we are mostly using it simply as something concrete to fix attention on.
The energy of the state with wavefunction $\psi$ is
\begin{equation}
  \label{eq:hA}
\frm{h}_{\bm A}[\psi] = \int |(\nabla - i{\bm A})\psi|^2\, dx,
\end{equation}
well-defined as a real-valued and lower-bounded
quadratic form defined on a dense subspace of $L^2(\Real^3)$.
Under a technical condition,
$\frm{h}_{\bm A}$ is naturally associated to a corresponding self-adjoint
operator $\frm{H}_{\bm A}$. This well-known theory is recovered as part of
the development in Section~\ref{sec:families}.
The resolvent $\Rmap(\zeta,H_{\bm A}) = (H_{\bm A}-\zeta)^{-1}$
is $\Cmplx$-analytic in $\zeta$. Is it also analytic in ${\bm A}$?
Merely posing the question shows that we should take ${\bm A}$ in
some {\em complex} space, and therefore allow the field ${\bm A}(x)$ to
be complex-valued.
Thus, we generalize (\ref{eq:hA}) to the sesquilinear form (\sqf)
\begin{equation}
  \label{eq:hA-c}
{\frm{h}}_{\bm A}[\phi,\psi]
=
\int ({\nabla} + i{\bm A})\overline{\phi} \cdot ({\nabla} - i{\bm A})\psi   \, dx,
\end{equation}
the diagonal part \hbox{$\frm{h}_{\bm A}[\psi] \defeq \frm{h}_{\bm A}[\psi,\psi]$}
being the associated quadratic form. We were careful to not have the complex
conjugate of ${\bm A}$ appear in the form (\ref{eq:hA-c}).
Now, we might ask whether
\hbox{$(\zeta,{\bm A}) \mapsto \Rmap(\zeta,{\bm A})$}
is a \hbox{$\Cmplx$-differentiable}, or {\it holomorphic}, map from some open subset of
\hbox{$\Cmplx \times \vec{L}^3(\Real^3)$} into $\Lin(\sH)$
(The arrow on $\vec{L}$ merely indicates a vector, rather than scalar, field).
Just as in elementary complex analysis, this is enough to guarantee
\hbox{$\Cmplx$-analyticity}, and then $\Real$-analyticity for real ${\bm A}$
by restriction. This theory of holomorphy in Banach spaces is reviewed in
Section~\ref{sec:Banach-holomorphy}.
A resolvent operator holomorphic in this sense is not an end in itself.
With its aid, however, one can show that isolated eigenvalues and properties
of the associated eigenvectors are analytic functions of the Hamiltonian
paramater, i.e., ${\bm A}$ in this case. This kind of application is
considered in detail in Section~\ref{sec:eigenvalues}.

There are certainly limits to how ${\bm A}$ can be allowed to vary
and have all this work. One might imagine that ${\bm A}$ should represent a
``not too big'' perturbation of $\frm{h}_{\bm 0}$.
A good idea, which we shall follow, is that all the {\sqf}s
$\frm{h}_{\bm A}$ should be mutually relatively bounded.
This suggests an abstract study of complete equivalence classes of mutually relatively
bounded {\sqf}s on a dense subspace $\sK$ of $\sH$, independently of
any concrete paramaterization (such as provided here by ${\bm A}$),
an idea which turns out to be quite fruitful.
As described below, the fundamental $\Rmap$-map $(\zeta,\frm{h}_x)\mapsto (H_x-\zeta)^{-1}$
and $\Emap$-map $(\beta,\frm{h}_x)\mapsto e^{-\beta H_x}$, which are the basis of applications,
are shown holomorphic at this abstract level.
We therefore know that a concrete parameterization enjoys these holomorphy properties as
soon as it is shown to parameterize part of such an equivalence class $\calC$
in the right way, and that turns out to be surprisingly easy.

The relation $\relbd$ of relative boundedness induces equivalence classes
of {\sqf}s on $\sK$. Suppose $\calC$ is one such (technically: containing
some {\it closable sectorial} form). 
When ${\bm A}$ is allowed to be complex, $\frm{h}_{\bm A}[\psi]$ is
no longer real, but appropriate restrictions
ensure that it takes values in a right-facing wedge in $\Cmplx$ as
$\psi$ varies among unit vectors in its domain.
This property is {\it sectoriality}, and the useful generalization of
lower-bounded self-adjointness
which allows complex analytic methods to be brought into play.
The class $\calC_\relbd$ of {\em all} {\sqf}s bounded relative
to $\calC$ has a natural Banach space structure, up to norm equivalence.
In fact, it can be identified with $\Lin(\sHp;\sHm)$, where $\sHp\subset\sH\subset\sHm$
is a Hilbert rigging of the ambient Hilbert space $\sH$.
Such Hilbert riggings play a very important role in our methodology, and
are reviewed in Section~\ref{sec:Hilbert-rigging-abstract}.

The sectorial forms in $\calC$, denoted $\sct{\calC}$, comprise an open subset
of $\calC_\relbd$, and are the ones of real interest.
They induce closed sectorial operators, $H$ corresponding to $\frm{h}$.
A central result is that $\frm{h} \mapsto H^{-1}$ is holomorphic 
from an open subset of $\sct{\calC}$ into $\Lin(\sH)$.
Note that the same cannot be said of $\frm{h} \mapsto H$, since the operators have
differing domains, hence it is not even clear in what Banach space we can locate them all. 
This result can be unfolded to display the resolvent explicitly since
$\sct{\calC}$ is invariant under translation by multiples of the identity:
$(\zeta,\frm{h}) \mapsto \Rmap(\zeta,H)$ is holomorphic on its natural domain in
$\Cmplx\times\sct{\calC}$.
The other central abstract result is holomorphy of 
$\frm{h} \mapsto e^{-H}$, as a map into $\Lin(\sH)$.
Again, we can unfold this to the map $(\beta,\frm{h})\mapsto e^{-\beta H}$
from a natural domain in $\Cmplx_{\text{rt}}\times \sct{\calC}$, where
$\Cmplx_{\text{rt}}$ is the open right half-plane. 

Returning to the vector potential ${\bm A}$, what needs to be done once
these abstract results are in place is very simple, as the theory summarized in
Section \ref{sec:Banach-holomorphy} shows.
Indeed, once an appropriate equivalence class of {\sqf}s is identified --- 
for instance, those equivalent to $\frm{h}_{\bm 0}$ on $C^\infty_c(\Real^3)$ ---
one only needs to check local boundedness and holomorphy on one-complex-dimensional
affine slices, and the latter really does reduce simply to noting that the
expression contains two ${\bm A}$'s.
The conditions seem easy to check in general. We can even anticipate that for a reasonable
family of closable sectorial {\sqf}s parameterized over an open subset $\calU$
of a Banach space, if they are all mutually relatively bounded, hence lie in some $\sct{\calC}$,
then this map $\calU\to\sct{\calC}$ will be holomorphic, and therefore so will be 
the $\Rmap$-map and $\Emap$-map.
Some uses of these are treated in Sections \ref{sec:eigenvalues} and
\ref{sec:free-energy}, respectively. The former can be used to show
analyticity of isolated eigenvalues, as well as various properties
of the associated eigenstates, such as charge and current density.
The negative exponential $e^{-\beta H}$ shows up in quantum physics in
two major contexts. One is analytic continuation of time, a topic which will not
be addressed here. The other is quantum statistical mechanics, where it represents
the thermal statistical operator, if it is trace-class.
This is dealt with in Section \ref{sec:free-energy}, where we show that,
if the free energy is (well-defined and) locally bounded, it is analytic,
and give conditions for that to be the case. 

\subsection{Organization of the paper}

Section~\ref{sec:Banach-holomorphy} provides needed background on
analyticity and holomorphy in Banach spaces.
Section~\ref{sec:Hilbert-riggings} provides background on Hilbert rigging.
Prop.~\ref{prop:iso-to-closed} is a nonstandard result there which
plays an important r\^ole in the later development.
Section~\ref{sec:families} is the technical core of the paper.
It develops the Banach space structure associated
with equivalence classes of closable {\sqf}s and
the general idea of {\RSF}s, and proves holomorphy of the $\Rmap$-map
in Thm.~\ref{thm:resolvent-holo}.
Section~\ref{sec:QM} is concerned with identifying specific holomorphic families of
nonrelativistic Hamiltonians --- magnetic Schr\"odinger forms parameterized by both
scalar and vector potentials. 
Sections~\ref{sec:eigenvalues} and \ref{sec:free-energy} on the other hand, are
concerned with what can be done if one has such a family, that is, what other
quatities inherit the holomorphy.
Section~\ref{sec:eigenvalues} studies low-energy Hamiltonians and
eigenstate perturbation. Holomorphy of the energy and charge/current densities
for isolated nondegenerate eigenstates is derived here, among other things.
Section~\ref{sec:free-energy} is concerned with holomorphy of the $\Emap$-map
and its consequences.
Under appropriate conditions, this yields holomorphy of the nonzero-temperature
statistical operator in {\em trace-norm}, as well of free energy and thermal expectations.
Special attention is again given to charge/current density.
Section~\ref{sec:summary} gives a selective summary.

\subsection{Conventions and notations}

For convenient reference, some conventions will be listed here.
$\sX$, $\sY$ and $\sZ$ denote generic Banach spaces, $\calU$
an open subset of a Banach space, $\sH$ denotes a Hilbert space.
$\Lin(\sX;\sY)$ is the space of
bounded linear operators from $\sX$ to $\sY$, with the usual operator norm,
$\Lin(\sX) = \Lin(\sX;\sX)$, $\Lin^1(\sH)$ and $\Lin^2(\sH)$ denote the spaces of
trace-class and Hilbert-Schmidt operators, respectively, and
$\Lincl(\sX)$ the set of densely-defined closed operators in $\sX$,
and $\Linv(\sX;\sY)$ that of invertible bounded operators (Banach isomorphisms).
Product spaces, e.g., $\sX\times\sY$ are usually denoted that way
rather than as $\sX\oplus\sY$ because the product notion matches the informal
interpretation better.
$\Rmap(z,A) = (A-z)^{-1}$ is the resolvent operator
(the notation $\Rmap$ will be overloaded later), $\res A$ the resolvent set,
and $\spec A$ the spectrum of $A$.
Topological closure is generally denoted by $\cl$, instead of an overbar.
Barred arrows specify functions, while plain arrows display the domain and
codomain, e.g., $A\mapsto e^{A}\colon \Lin(\sH) \to \Lin(\sH)$ is exponentiation
on bounded operators. $(x_n)_{n\in\Nat}$ or $(x_n)_n$, or even just $(x_n)$ if it
is unambiguous, denotes a sequence.
Additional notations will be defined as need arises.
Definitions and all theorem-like environments share a common counter;
the numbering is merely a navigational aid.

\section{Analyticity in the Banach space setting}
\label{sec:Banach-holomorphy}

This section reviews the necessary theory of differential calculus and
holomorphy in Banach spaces.
The material on differential calculus reviewed in Section~\ref{sec:calculus}
is quite standard and can be found in many
places\cite{Lang-Real_analysis,AMR,AMP,Chae}.
The theory of holomorphy in Banach spaces discussed in Section~\ref{sec:holomorphy},
much less so. In depth treatments are in monographs of Mujica\cite{Mujica}
and Chae\cite{Chae}. Thm.~\ref{thm:holomorphy-summary} is the reason this section
is here at all, and the other results are mostly concerned with making the
demonstration of holomorphy as easy as possible.

\subsection{Differential calculus}
\label{sec:calculus}

In this subsection, the base space of Banach spaces ($\sX$, $\sY$, \dots)
may be either $\Real$ or $\Cmplx$. $\calU$ is an open subset of
a Banach space (usually $\sX$).

\subsubsection{Derivatives}

If $\Arr{\calU}{f}{\sY}$ admits a linear approximation near $a\in \calU$ as
\begin{equation}
f(a+x) = f(a) + Df(a) x + o(\|x\|),
\end{equation}
for some continuous linear map
$\Arr{\sX}{Df(a)}{\sY}$, i.e. $Df(a) \in \Lin(\sX;\sY)$,
then $Df(a)$ is said to be the {\it Fr\'{e}chet differential} (or {\it derivative})
of $f$ at $a$.

We are not interested in differentiability at isolated points only, but 
throughout $\calU$. $f$ is $C^1$ on $\calU$ if $Df$ is everywhere defined on $\calU$
and continuous.
In that case, $\Arr{\calU}{Df}{\Lin(\sX;\sY)}$ is itself a continuous map into
a Banach space (with the usual operator norm) and we may ask about differentiability of $Df$.

If the differential of $Df$ at $a$, denoted $D^2f(a)$, exists it belongs to
$\Lin(\sX,\Lin(\sX;\sY))$, by definition.
Thus, for $x,x'\in \sX$, $D^2f(a)(x) \in \Lin(\sX;\sY)$ 
(dropping some parentheses and writing simply `$D^2f(a)\, x$' is a good idea)
and $D^2f(a)\, x\, x' \in \sY$.
Elements of $\Lin(\sX;\Lin(\sX;\sY))$ are actually
{\em bilinear}, that is, linear in each argument with the other held fixed.
Moreover, $D^2f(a)$ is symmetric, that is, $D^2f(a)\, x\, x' = D^2f(a)\, x'\, x$.
This symmetry continues to higher orders, as long as differentiability holds,
and provides good motivation to think primarily in terms of multilinear mappings
rather than nested linear mappings.
Eliding the distinction between a nested operator in
$\Lin(\sX;\cdots \Lin(\sX;\sY)\cdots )$,
the $n$-th differential $D^nf(a)\in\Lin(\sX,\ldots,\sX;\sY)$ is a continuous,
symmetric, $n$-linear map from $\sX\times\cdots \times \sX$ into $\sY$.

\subsubsection{Taylor series and analyticity}

If $f$ is continuously differentiable, then whenever the line segment from
$a$ to $a+x$ is in $\calU$,
\hbox{$f(a+x) = f(a) + (\int_0^1 Df(a+tx) \, dt)\, x$}.
Suspending the question of convergence, one deduces that the Taylor series
expansion should be $\sum_{n=0}^\infty \frac{1}{n!} D^nf(a)\, x\cdots x$.
If, for every point $a\in \calU$, the Taylor series expansion of $f$ converges to $f$
uniformly and absolutely on a ball of some nonzero ($a$-dependent) radius,
$f$ is said to be {\it analytic} on $\calU$. This is the favorable situation in
which we are interested. The notion of analyticity, and the actual use of a convergent
series expansion, is independent of the base field, but it follows from a {\it prima facie}
much weaker condition when the base field is $\Cmplx$, as discussed next.

\subsection{Holomorphy}
\label{sec:holomorphy}

This subsection is concerned with the equivalence between holomorphy and
$\Cmplx$-analyticity (Thm.~\ref{thm:holomorphy-summary}) and ways to make
the demonstration of holomorphy easy (nearly everything else).

\subsubsection{Complex linearity and conjugate linearity}
\label{sec:C-linearity}

A function of type ${\Real}\rightarrow{\Real}$
can be differentiable to all orders without being analytic, whereas the situation
is remarkably otherwise for those of type $\Cmplx \rightarrow \Cmplx$.
What is much less well-appreciated is that this contrast persists even in infinite-dimensional
Banach spaces. 
Now we assume that the base field for $\sX$ and $\sY$ is $\Cmplx$.
They can still be regarded as real vector spaces $\sX_{\Real}$, $\sY_{\Real}$ by restriction
of scalars; in that case $ix$ is considered not a scalar multiple of $x$, but a vector
in an entirely different ``direction''. Suppose
$\Arr{\sX_{\Real}}{f}{\sY_{\Real}}$ is $\Real$-differentiable at $a$,
temporarily denote the differential as $D_\Real f(a)$, and define
$Df(a)(x) = \frac{1}{2}[D_\Real f(a)(x) - i D_\Real f(a)(ix)]$,
$\overline{D}f(a)(x) = \frac{1}{2}[D_\Real f(a)(x) + i D_\Real f(a)(ix)]$.
The condition for $\Cmplx$-differentiability is then $\overline{D}f(a) = 0$.
This is the analog of the Cauchy-Riemann equation.
The function $f$ is said to be {\it holomorphic} on $\calU$ if it is
\hbox{$\Cmplx$-differentiable}
there.
Sometimes (e.g., Chae\cite{Chae})
{\it holomorphic} is instead taken synonymous with $\Cmplx$-analyticity by definition,
but it does not really matter as the following remarkable theorem shows.
\begin{thm}
\label{thm:holomorphy-summary}
For {\em complex} Banach spaces $\sX$ and $\sY$,
$\calU$ open in $\sX$, the following properties of
$\Arr{\calU}{f}{ \sY}$ are equivalent:
  \newline
\textnormal{(a)}
holomorphy \textnormal{(}$\Cmplx$-differentiability\textnormal{)}
  \newline
\textnormal{(b)}
infinte $\Cmplx$-differentiability
  \newline
\textnormal{(c)}
$\Cmplx$-analyticity
\end{thm}
\begin{proof}
  See \S\S 8 and 14 of Mujica\cite{Mujica}; or
  Chae\cite{Chae}, Thm.~14.13.
\end{proof}

Thus, even if we are ultimately interested only in $\Real$-analyticity in some real subspace
$\tilde{\sX}$ of a complex space $\sX$, it can be advantageous to work in $\sX$, establish
holomorphy (a comparatively simple property) to get $\Cmplx$-analyticity in $\sX$ and
thence $\Real$-analyticity in $\tilde{\sX}$ by restriction.
This is in the spirit of Jacques Hadamard's famous dictum,
``{Le plus court chemin entre deux v\'{e}rit\'{e}s dans le domaine r\'{e}el
  passe par le domaine complexe}''.

The following mostly simple permanence properties of holomorphy are important:
\begin{itemize}
\item Composition:
Whenever \hbox{$\Arr{\sX\supset\calU}{f}{\sY}$}
  and \hbox{$\Arr{\sY\supset {\mathcal V}}{g}{\sZ}$} are holomorphic, so is
\hbox{$\Arr{\calU \cap f^{-1}({\mathcal V})}{g\circ f}{\sZ}$}. 
\item Inversion:
 \hbox{$\Lin_{\text{iso}}(\sX;\sY) \overset{\mathrm{inv}}{\to}  \Lin_{\text{iso}}(\sY;\sX)$}
 is holomorphic, where $\Lin_{\text{iso}}(\sX;\sY)$
 denotes the open set of invertible operators in $\Lin(\sX;\sY)$.
\item Products:
  If the domain of $f$ is in a product space $\sX_1\times \sX_2$, then
$f$ is holomorphic iff it is jointly continuous and separately holomorphic.
\item Equivalent norms:
  Holomorphy is stable under equivalent renorming of the domain or codomain space.
\item Differentiation: If $f$ is holomorphic, so is $Df$.
\item Sequential limits: Sequential convergence uniformly on compact sets
  preserves holomorphy. (See Prop. \ref{prop:convergent-sequences})
\end{itemize}

\subsubsection{Reduction to 1D domain or range}

The preceding shows that convergence of series expansions can be deduced
from the mere existence of a differential. However, the latter is still
complicated by the infinite-dimensional setting. Fortunately, a remarkable
reduction is possible here as well --- to consideration of one-$\Cmplx$-dimensional
subspaces both in the domain and codomain (together with local boundedness).
\begin{defn}
  \label{def:G-holo}
$\Arr{\calU}{f}{\sY}$ is {\it G-holomorphic} if
for all $x\in \calU$, $y\in\sX$,
\hbox{$\zeta \mapsto f(x+\zeta y)$} 
is an ordinary holomorphic function of $\zeta$ on some
neighborhood of zero in $\Cmplx$.
\end{defn}

The following fundamental theorem is named after Graves, Taylor,
Hille and Zorn\cite{Chae}.
\begin{thm}[GTHZ]
  \label{thm:GTHZ}
  For a map $\Arr{\calU}{f}{\sY}$, the following property equivalence holds:
  \newline
  holomorphy $\Leftrightarrow$ G-holomorphy and locally boundedness.  
\end{thm}
\begin{proof}
See Mujica, Prop.~8.6 and Thm.~8.7;
Chae\cite{Chae}, Thm.~14.9.
\end{proof}
\begin{rem}
  By defintion, {\it locally bounded} means bounded on some neighborhood of each
  point of the domain. In a Banach space (or even a metric space), this is equivalent
  to boundedness on compact subsets of the domain.
\end{rem}

Maps into spaces of linear operators will be very important in the following
and are considered now. In fact, since any Banach space is isometrically embedded
in its bidual, this is not really a special case.

\begin{defn}
\label{def:wk-st-wo-holo}
$\Arr{\calU}{f}{\sY}$ is {\em weakly holomorphic} if
\hbox{$x \mapsto \pair{\lambda}{f(x)}\in \Cmplx$} is holomorphic for each $\lambda\in\sY^*$.
It is {\em densely weakly holomorphic} if the condition holds for a set of $\lambda$'s
dense in $\sY^*$.

$\Arr{\calU}{f}{\Lin(\sY;\sZ)}$ is
{\em strongly holomorphic} if
\hbox{$x \mapsto f(x)y \in \sZ$} is holomorphic for each $y\in\sY$;
and {\em weak-operator holomorphic} if
\hbox{$x \mapsto \pair{\lambda}{f(x)y}\in \Cmplx$} is holomorphic for each
$y\in\sY$ and $\lambda\in\sZ^*$.
As for weak holomorphy, these may be modified with {\em dense} to indicate that
the set of $y$'s [resp. pairs $(y,\lambda)$] in question is dense in $\sY$ [resp. $\sY\times \sZ^*$].
\end{defn}

The following Lemma is preparation for
Propositions \ref{prop:wk-holo} and \ref{prop:WO-holo}.
Some obvious abbreviations (`st.' for `strong', `loc. bdd.' for `locally bounded',
`holo' for `holomorphic') are used.
\begin{lem}
  \label{lem:st-holo}
For \hbox{$\Arr{\calU}{f}{\Lin(\sY;\sZ)}$}, the following property implications hold.
  \newline
  \textnormal{(a)}
st. G-holo. $\Rightarrow$ G-holo.
  \newline
\textnormal{(b)}
loc. bdd. \& dense st. holo. $\Rightarrow$ st. holo.
  \newline
\textnormal{(c)}
st. holo. $\Rightarrow$ loc. bdd.
\newline
\textnormal{(d)}
loc. bdd. \& dense st. G-holo. $\Rightarrow$ holo.
\newline
\textnormal{(e)}
st. holo. $\Rightarrow$ holo.
\end{lem}

\begin{rem}
Parts (a), (b), and (c) are really just preparation for (d) and (e).
\end{rem}
  
\begin{proof}
\noindent (a): 
Since G-holomorphy concerns affine planes independently, assume that $\calU \subseteq \Cmplx$ 
without loss.
Assume (to be justified later) that $f$ is also continuous. Then,
for every $y\in \sY$ and simple closed contour $\Gamma$ in $\calU$,
  \begin{equation}
0 = \oint_\Gamma f(\omega) y \frac{d\omega}{2\pi i}
 = \left[ \oint_\Gamma f(\omega) \frac{d\omega}{2\pi i} \right] y.
  \end{equation}
Continuity of $f$ is used here to justify taking $y$ outside the integral.
Since $y$ ranges over $\sY$, which is separating for $\Lin(\sY;\sZ)$,
the integral in square brackets is zero.
Finally, Morera's theorem implies that $f$ is holomorphic, because $\Gamma$ is arbitrary.

To complete the proof of (a), we must show that $f$ is continuous at
$\zeta\in\calU$. Suppose not. Then there is a sequence $\calU\ni \zeta_n \to \zeta$
such that \hbox{$\|f(\zeta') - f(\zeta)\|/(\zeta_n-\zeta) \to \infty$}, and
by the uniform boundedness principle, $y\in\sY$ such that
\hbox{$(f(\zeta')y - f(\zeta)y)/(\zeta_n-\zeta)$} diverges. However, since $f$ is
strongly holomorphic, the limit of the latter is \hbox{$\frac{d}{dz}(f(z)y)|_{z=\zeta}$}.
Contradiction.
  \smallskip
  
  \noindent (b):
We need to show that, for each $y\in \sY$, $x \mapsto f(x) y$ [abbreviated here
$f(\;) y$] is holomorphic near each point of $\calU$.
By the dense strong holomorphy assumption, there
is $D$ dense in $\sY$ such that, for every $u\in D$, $f(\;) u$ is holomorphic.
Also, for any sequence $D \ni y_n \to y$, local boundedness implies that
the sequence $f(\;) y_n$ converges not merely pointwise, but locally uniformly, 
to $f(\;) y$, which is therefore holomorphic by Prop. \ref{prop:convergent-sequences}.
  
  \noindent (c): 
Fix compact $K\subset \calU$. For every $y\in\sY$, 
${f(\;)y}$ is holomorphic by hypothesis, therefore continuous, therefore bounded on $K$.
The uniform boundedness principle secures boundedness $f$ on $K$.

\noindent (d): 
local boundedness \& dense strong G-holomorphy implies strong G-holomorphy by
the G\^ateaux version of (b),
which implies G-holomorphy by (a).
Finally, holomorphy follows by Thm. \ref{thm:GTHZ}.

\noindent (e):
We have G-holomorphy by (a), and local boundedness by (c).
Again, conclude via Thm. \ref{thm:GTHZ}.
\end{proof}

\begin{prop}
  \label{prop:wk*-holo}
  For $\Arr{\calU}{f}{\sY^*}$, the following are equivalent:
  \newline
  \noindent (a) holomorphy
  \newline
  \noindent (b) weak-* holomorphy
  \newline
  \noindent (c) local boundedness \& dense weak-* G-holomorphy
\end{prop}
\begin{proof}
  This follows immediately from Lemma~\ref{lem:st-holo} for the case
  $\sY^* \simeq \Lin(\sY;\Cmplx)$, realizing that the adjective
  ``strong'' there specializes to ``weak-*''.
\end{proof}
\begin{prop}
  \label{prop:wk-holo}
  For $\Arr{\calU}{f}{\sY}$, the following are equivalent:
  \newline
  \noindent (a) holomorphy
  \newline
  \noindent (b) weak holomorphy
  \newline
  \noindent (c) local boundedness \& dense weak G-holomorphy
\end{prop}
\begin{proof}
  ${\sY}$ is isometrically imbedded in its bidual 
  \hbox{${\sY}^{**} \cong {\Lin(\sY^*;\Cmplx)}$.}
Now apply Prop.~\ref{prop:wk*-holo}.  
\end{proof}
\begin{prop}
  \label{prop:WO-holo}
  For $\Arr{\calU}{f}{\Lin(\sY;\sZ)}$, the following are equivalent:
  \newline
  \noindent (a) holomorphy
  \newline
  \noindent (b) weak-operator holomorphy
  \newline
  \noindent (c) loc. bdd. \& dense weak-operator G-holomorphy
\end{prop}
\begin{proof}
  Use the same trick as in Prop.~\ref{prop:wk-holo} to write
\hbox{$\Arr{\calU}{f}{ \Lin(\sY;\Lin(\sZ^*;\Cmplx)) }$},
apply Lemma~\ref{lem:st-holo} directly, and then Prop.~\ref{prop:wk-holo}.
\end{proof}

Although holomorphy for the case $\sX\equiv \sY\equiv \Cmplx$ is not usually discussed in
terms of linear operators as here, we may note that it fits in perfectly.
The operator $Df(a)$ in that case can be construed simply as multiplication
by a complex number, $\partial f(a)$, so that $a\mapsto Df(a)$ is identified
with the complex function $a \mapsto \partial f(a)$.
Differentiation does not generate objects of a fundamentally different
type in that case. For higher-dimensional Banach spaces, however, it does so,
and part (b) of Thm.~\ref{thm:holomorphy-summary} thereby gains in importance.
The $D^nf$, as $n$ varies, all have distinct codomains, yet they are
all holomorphic if $f$ is so.

We close this Section with a proof of the sequential permanence property
mentioned earlier, which is also found as Prop. 9.13 of Mujica.
\begin{prop}
  \label{prop:convergent-sequences}
  If $\Arr{\calU}{f_n}{\sY}$ is a sequence of holomorphic mappings converging
  to $f$ uniformly on compact subsets of $\calU$, then $f$ is holomorphic.
\end{prop}
\begin{proof}
Use Thm.~\ref{thm:GTHZ}
(\hbox{G-holomorphic} and locally bounded $\Leftrightarrow$ holomorphic).
For any compact subset $K$ of $\calU$,
the $f_n$'s are bounded, and converge uniformly to $f$, hence $f$ is bounded.
By Prop.~\ref{prop:wk-holo} G-holomorphy of $f$ reduces to the
case \hbox{$\calU\subseteq \sY = \Cmplx$}, which is a well-known result 
of classical complex analysis.
\end{proof}

\section{Hilbert riggings}
\label{sec:Hilbert-riggings}

This section is also primarily background, although Prop.~\ref{prop:iso-to-closed}
is not standard and will play an important r\^ole.
Section~\ref{sec:kinetic-energy} is a concrete illustation of Hilbert rigging intended
primarily for those unfamiliar with the idea. 
A Hilbert rigging of a Hilbert space $\sH$ is a sandwiching
$\sHp\subset \sH \subset\sHm$ by two other Hilbert spaces such that
$\sHm$ is the dual space of $\sHp$ with respect to the original inner product on $\sH$.
They will be used through the identification
of a family $\calC_\relbd$ of {\sqf}s in $\sH$ with $\Lin(\sHp;\sHm)$ for an appropriate $\sHp$. 
Prop.~\ref{prop:iso-to-closed} concerns the identification of isomorphisms from
$\sHp$ to $\sHm$ with closed operators on $\sH$. 

\subsection{Example: kinetic energy}
\label{sec:kinetic-energy}

Before presenting the abstract construction of Hilbert rigging,
we illustrate briefly with the concrete and pertinent example of kinetic energy.
The reader unfamiliar with Hilbert riggings may find it helpful to keep this example
in mind in Section~\ref{sec:Hilbert-rigging-abstract}.

Thus, take $\sHz$ to be $L^2(\Real^n)$; the inner product is
\begin{equation}
  \label{eq:H0}
\inpr{u}{v}
 = \int u(x)^* v(x) \,  d^nx
 = \int \widetilde{u}(p)^* \widetilde{v}(p) \, d^np.
\end{equation}
Fourier transform will be indicated (in this subsection only) by an over-tilde, as above.

Now, a sesquilinear form corresponding to kinetic energy is
\begin{equation}
  \label{eq:KE-form}
  \inpr{\phi}{\psi}_+
 \defeq
    \inpr{\phi}{\psi} +
    \sum_{i=1}^{n}\inpr{\partial_i \phi}{\partial_i \psi}.
  \end{equation}
  To be precise,
  $\|{\psi}\|_+^2 = \inpr{\psi}{\psi}_+$ is the kinetic energy of vector
  state $\psi$, up to the addition of $\|{\psi}\|^2$.
  The notation suggests, as indeed is the case, that this sesquilinear form is
  a legitimate inner product. Moreover, it corresponds to a Hilbert space $\sHp$
  based on a dense subspace of $\sH$. That this is so is best seen in momentum
  space, a move which also alleviates the technical compication that
  we must be careful to {\it a priori} interpret
  the derivatives in (\ref{eq:KE-form}) in a weak or distributional sense.
  The momentum space expression is
\begin{equation}
\label{eq:KE-form-momentum}
\inpr{\phi}{\psi}_+
 = \int \widetilde{\phi}(p)^* \widetilde{\psi}(p) \,  (1+|p|^2) d^np.
\end{equation}
This clarifies both that there really is a subspace of $\sH$ which is complete for  
the new inner product $\inpr{\phantom{\phi}}{\phantom{\psi}}_+$ and why we included
the term $\inpr{{\phi}}{{\psi}}$ in (\ref{eq:KE-form}).

Authorized by the Riesz-Fr\'echet theorem, we could identify $\sHp$ with its dual
space as usual, associating $\phi\in\sH_+$ with the functional
$\psi \mapsto \inpr{\phi}{\psi}_+$.
However, we want to identify the dual with respect not to 
$\inpr{\phantom{\phi}}{\phantom{\psi}}_+$, but with respect to
$\inpr{\phantom{\phi}}{\phantom{\psi}}$.
The momentum-space expression (\ref{eq:KE-form-momentum}) makes clear 
how to do this:
Define $J$ by $\widetilde{J\phi}(p) = (1+|p|^2)\widetilde{\phi}(p)$, so
that $\inpr{\phi}{\psi}_+ = \inpr{J\phi}{\psi}$, where the last represents
some extension of the inner product on $\sH$. With the inner product
\begin{equation}
\label{eq:H-minus}
\inpr{\phi}{\psi}_{-}
 = \int \widetilde{\phi}(p)^* \widetilde{\psi}(p) \,  (1+|p|^2)^{-1} d^np,
\end{equation}
we get another Hilbert space $\sHm$ such that \hbox{$\sHp\subset\sH\subset\sHm$},
and $\Arr{\sHp}{J}{\sHm}$ is unitary.
All three of these spaces consist of functions in momentum space,
but elements of $\sHm$ are actually tempered distributions, in general.
For instance, if $n=1$, $\sHm$ contains delta-functions.
Now we can clarify the meaning of $\inpr{J\phi}{\psi}$:
The map $\sHp\times\sHp \ni (\phi,\psi) \mapsto \inpr{\phi}{\psi}$
admits an extension by continuity to either $\sH$ in both factors
(yielding the ordinary inner product), or to $\sHm$ in one factor.

\subsection{General construction}
\label{sec:Hilbert-rigging-abstract}

We now review the abstract idea of a {\it Hilbert rigging}
as summarized in the (not commutative!) diagram
\begin{equation}
  \label{eq:scale-of-spaces}
  \begin{tikzcd}
    \sHp \arrow[r,hookrightarrow,"\iota\subsub{+}"]
    \arrow[rr,bend right=35, "J"]  
    &
    \sHz \arrow[r,hookrightarrow,"\iota\subsub{0}"]
    & \sHm
    \arrow[ll,bend right=35, "J^{-1}"']  
 \end{tikzcd}
 \end{equation}
Expositions of this technology can be found in
\S II.2 of Simon\cite{Simon-Forms},
\S VIII.6 of Reed~\&~Simon\cite{Reed+Simon},
Ch.~4 of de~Oliveira\cite{deOliveira},
or \S 14.1 of Berezansky\cite{Berezansky-II}.

Start with a Hilbert space $\sHz$ with inner produce $\inpr{\phantom{u}}{\phantom{v}}$,
and a dense subspace equipped with stronger inner product $\inpr{\phantom{u}}{\phantom{v}}_+$,
which makes it into a Hilbert space $\sHp$, so that the inclusion
of one underlying vector space $\{\sHp\}$ into the other $\{\sH\}$
induces a continuous injection
$\iota_{\scriptscriptstyle{+}} \colon \sHp {\hookrightarrow} \sHz$.

The adjoint of $\iota_{_+}$, defined by
\begin{equation}
  \label{eq:i*}
  \ilinpr{\iota_{_+}^* u}{\psi}_{+} =  \inpr{u}{\iota_{_+}\psi}
\end{equation}
is also injective with dense image, since taking adjoints swaps those properties.
Use $\iota_{_+}^*$ to define a new inner product on $\{\sHz\}$ via
\begin{equation}
\inpr{u}{v}_{-} \defeq \ilinpr{\iota_{_+}^* u}{\iota_{_+}^* v}_{_+},  
\end{equation}
equipped with which it becomes the preHilbert space $\{\sHz\}_{-}$,
with a completion denoted $\sHm$. The inclusion of $\{\sHz\}$ into
$\sHm$ is $\iota_0$. By construction, $\iota_{_+}^*$ extends by continuity to
a unitary mapping
\begin{equation}
J^{-1}\colon \sHm\overset{\sim}{\to} \sHp.
\end{equation}
Thus, suppressing the injection $\iota_+$ of $\sHp$ into $\sH$, we may
rewrite (\ref{eq:i*}) as
\begin{equation}
  \label{eq:dual-with-respect-to-H}
  \inpr{u}{\psi} 
  = \ilinpr{J^{-1} u}{\psi}_{+}.
\end{equation}
Furthermore, according to the preceding, the right-hand side extends by continuity to a
continuous sesquilinear map on $\sHm\times \sHp$ with $J^{-1}\sHm = \sHp$.
Using (\ref{eq:dual-with-respect-to-H}) then to define an extension of the $\sH$ inner product
$\inpr{\;}{\;}$ to $\sHm\times\sHp$, we say that $\sHm$ realizes the dual
space of $\sHp$ relative to the original inner product.

%

The maps in (\ref{eq:scale-of-spaces}) naturally induce
two {\em bounded} linear mappings
\begin{align}
\nonumber
  &
    { T \mapsto \iota_0 T \iota_+ }\;\colon
    {\Lin(\sH)}\to {\Lin(\sHp;\sHm)},
            \nonumber  \\
&  {T \mapsto \iota_+ T \iota_0}\;\colon {\Lin(\sHm;\sHp)}\to {\Lin(\sH) }.
            \nonumber
\end{align}
These will be useful below. More interesting, though, is
a map that takes arbitrary $\hat{T}\in\Lin(\sHp;\sHm)$ into 
a (generally unbounded) linear operator $T$ on $\sH$ according to the following
notational convention.
\begin{cnvntn}
  \label{cnvntn:hats}
For \hbox{$\hat{T}\in\Lin(\sHp;\sHm)$},
${T}$ denotes the restriction of $T$ to 
\hbox{$\dom {T} = \setof{\psi\in\sHp}{\hat{T}\psi\in\sH}$},
considered simply as an operator {\em in} $\sH$.
\end{cnvntn}
Not every linear operator in $\sH$ comes from an operator in $\Lin(\sHp;\sHm)$
in this way, so one should not think of the hat as a map or transform of some sort;
the map actually goes the other way.

The following Proposition can be viewed as an analog of
Lemma~\ref{lem:inverse-of-perturbed-closed-op}.
It plays an important r\^ole in the theory.
\begin{prop}
  \label{prop:iso-to-closed}
  Given \hbox{$\hat{T}\in\Linv(\sHp;\sHm)$}.
\newline\noindent\textnormal{(a)}
${T}\in\Lincl(\sH)$, i.e., it is closed with dense domain.
\newline\noindent\textnormal{(b)}
\hbox{$\Arr{\Linv(\sHp;\sHm)}{\hat{T}\mapsto {T}^{-1} }{\Lin(\sH)}$} is holomorphic.
\end{prop}
\begin{proof}
${T}^{-1} = {\iota_+}{\hat{T}^{-1}}{\iota_0}$
is bounded with domain $\sH$, hence closed, hence so is ${T}$.
Since $\iota_0$ and $\iota_+$ have dense image, $\dom{T}$ is dense in $\sH$.
Finally, $\hat{T}\mapsto {T}^{-1}$ is holomorphic since it is explicitly a composite of
inversion and composition with a linear map, which are holomorphic operations.
\end{proof}
Certainly ${T}^{-1}$ exists for some operators $\hat{T}$ in $\Lin(\sHp;\sHm)$
which are not invertible, and one may ask whether \hbox{$\hat{T}\mapsto{T}^{-1}$}
is holomorphic on a larger domain.
Close examination of this question is postponed to Prop.~\ref{prop:final-piece}
when more motivation will be in place.


\section{Families of forms and operators}
\label{sec:families}

This section is the technical core of the paper, preparing for applications in
Sections \ref{sec:QM}, \ref{sec:eigenvalues}, and \ref{sec:free-energy}.
Section~\ref{sec:sforms-1} recalls some basic ideas and definitions
connected with sesquilinear forms ({\sqf}s).
That is preparation for consideration of families of sectorial forms
parameterized over an open set $\calU$ of some Banach space. We want these
parameterizations to be holomorphic, hence the generalization of the
$\Real$-centered notion of lower-bounded hermitian to sectorial.
However, this can make sense only if relevant classes of {\sqf}s
have a Banach space structure themselves. 
Thm.~\ref{thm:[]+-} solves this problem, showing that the class $\calC_\relbd$ of
{\sqf}s relatively bounded with respect to an equivalence class $\calC$ of closable
forms is naturally identified with $\Lin(\sHp;\sHm)$, where $\sHp\subset \sH\subset \sHm$
is an Hilbert rigging. Attention is then turned to the closed operators associated with
the sectorial forms $\sct{\calC}$ in $\calC$.
Thm.~\ref{thm:resolvent-holo} is the second main result, showing that
the operator $H$ associated with $\frm{h}\in\sct{\calC}$ is invertible iff
$\frm{h}$ viewed as an element of $\Lin(\sHp;\sHm)$ is so.
This gives holomorphy of the $\Rmap$-map $(\zeta,\frm{h}) \mapsto (H-\zeta)^{-1}$
on its natural domain in $\Cmplx\times\sct{\calC}$, which will be a basic tool in
Sections \ref{sec:eigenvalues} and \ref{sec:free-energy}.
Attention then swings back to parameterizations and convenient criteria for
a family $\frm{h}$ of {\sqf}s to be a \RSF, i.e., holomorphically embedded in
some $\sct{\calC}$.

\subsection{Sesquilinear forms}
\label{sec:sforms-1}

This section consists mostly of definitions and notational conventions.
as well as some notational conventions. A standard source for this material
is \S\S VI.1,2 of Kato's treatise\cite{Kato}.

\begin{enumerate}[label={(\arabic*)}]
\label{defn:SQF}
\item 
A {\it sesquilinear form} ({\it \sqf} henceforth) $\frm{h}$ on complex vector space $\sK$
is a map \hbox{$(\phi,\psi) \mapsto \frm{h}[\phi,\psi] \Type{\sK\times \sK}{\Cmplx}$}
linear in the second variable and conjugate-linear in the first.
(Conjugate-linearity distinguishes these from bilinear forms.)
Dirac-style notation will also be used:
$\Dbraket{\phi}{\frm{h}}{\psi} \equiv \frm{h}[\phi,\psi]$.

To a sesquilinear form is associated a {\it quadratic form}
$\frm{h}[\psi] \defeq \frm{h}[\psi,\psi]$.
The sesquilinear form can be recovered by polarization, so we will
always use the term \sqf\ for economy.

We write $|\frm{t}|$ for the map $\psi\mapsto |\frm{t}[\psi]|$.
This {\em is not} an \sqf, unless $|\frm{t}| = \frm{t}$.

\item
  The {\it adjoint} of the {\sqf} $\frm{h}$ is
\hbox{$\frm{h}^*[\phi,\psi] \defeq \overline{\frm{h}[\psi,\phi]}$}.
If \hbox{$\frm{h} = \frm{h}^*$}, $\frm{h}$ is {\it hermitian}.
$\frm{h}$ is split into {\it real} and {\it imaginary} hermitian parts as
$\frm{h} = \frm{h}^r + i \frm{h}^i$ with
$\frm{h}^r = \frac{1}{2}(\frm{h} + \frm{h}^*)$, $\frm{h}^i = \frac{1}{2i}(\frm{h} - \frm{h}^*)$.
Hermitian quadratic forms are partially ordered similarly to self-adjoint operators:
$\frm{h} \le \frm{h}'$ means
$\forall\psi\in\sK, \; \frm{h}[\psi] \le \frm{h}'[\psi]$.
The inner product of the ambient Hilbert space provides the special
\sqf\ \hbox{${\bm 1}[\phi,\psi] \defeq \inpr{\phi}{\psi}$}.

\item 
The {\it numerical range} of $\frm{h}$ is the set 
\begin{equation}
\label{eq:numerical-range}
\Num \frm{h} \defeq \setof{\frm{h}[\psi]}{\psi\in\dom \frm{h}, \|\psi\|=1}.
\end{equation}
The role of numerical range for {\sqf}s somewhat analogous to that 
of {\it spectrum} for operators.
\begin{lem}
\label{lem:Num-cvx}
$\Num \frm{h}$ is a convex set.
\end{lem}
\begin{proof}
  We need to show that the line segment in $\Cmplx$ from
  $\frm{h}[\psi]$ to $\frm{h}[\phi]$ is in $\Num \frm{h}$, for
  unit vectors $\psi,\phi \in\dom \frm{h}$.
  By suitable scaling and translation (replace $\frm{h}$ by $a\frm{h}+b{\bm 1}$), we may assume
  that $\frm{h}[\psi]=0$ and $\frm{h}[\phi]=1$.

Define
\hbox{$\varphi(s) = (1-s)\psi + s e^{i\theta}\phi$}
for $0\le s\le 1$,
with $\theta$ to be chosen.
Then,
\begin{equation}
  \nonumber
 \frm{h}[\varphi(s)]
  = s^2 + s(1-s)\Big\{ e^{i\theta} \frm{h}[\psi,\phi] + e^{-i\theta} \frm{h}[\phi,\psi]\Big\}.
  \end{equation}
  For suitable choice of $\theta$, the quantity in braces, thus
  $\frm{h}[\varphi(s)]$ is real.
  $\frm{h}[\varphi(s)]$ goes continuously from
  $0$ to $1$ as $s$ increases from $0$ to $1$, and therefore covers at
  least the segment $[0,1]$.
  Since $\varphi(0)$ and $\varphi(1)$ are already normalized,
  normalizing $\varphi(s)$ will not alter this conclusion.
\end{proof}

\item \label{item:sector}
 An {\em open sector} is a right-facing wedge,
\begin{equation}
  \nonumber
\oSec{c}{\theta} \defeq \setof{c + r e^{i\varphi}}{r > 0,\, |\varphi| < \theta},
\end{equation}
in $\Cmplx$ for some {\em vertex} $c\in\Cmplx$ and {\em half-angle} $\theta < \pi/2$,
and the {\em closed sector} $\cSec{c}{\theta}$ is its closure.
If sector $\Sigma$ is contained in the interior of $\Sigma'$ and 
$\Sigma'$ has a strictly larger half-angle than does $\Sigma$,
then $\Sigma'$ is a {\em dilation} of $\Sigma$.

\item
\label{item:sectorial}
$\frm{h}$ is {\it sectorial} if its numerical range is contained in some sector,
and any such will be said to be {\it a sector for $\frm{h}$}.
$\Sigma$ is an {\em ample sector} for $\frm{h}$ if it is a dilation of
some sector for $\frm{h}$.

For any sectorial form $\frm{h}$, $\frm{h}^+$ will denote an arbitrary translate
$m{\bm 1} + \frm{h}^r$ such that ${\bm 1} \le \frm{h}^+$.
(Of course, the choice of $m$ can be standardized, but for our purposes there is no need.)

%
\item
  Any operator $T$ in $\sH$ naturally induces an \sqf\ on $\dom T$
  by $(\phi,\psi) \mapsto \inpr{\phi}{T\psi}$. The numerical range of
  $T$ is simply the numerical range of this \sqf.
  Caution: a closed operator is called {\it sectorial} if its spectrum lies
  in a sector. This is not the same thing as the associated \sqf\ being sectorial;
  the latter is a stronger condition. The relation between numerical range and
  spectrum is taken up in Section~\ref{sec:Num-and-spec}.
\item
The vector space of {\sqf}s on
a dense subspace $\sK$ of $\sH$ will be denoted $\SF(\sK)$.
The set of sectorial {\sqf}s on $\sK$, denoted $\sct{\SF}(\sK)$,
is a cone in $\SF(\sK)$. Generally, superscript `$\triangleleft$'
indicates the sectorial members of any class of {\sqf}s.

\item 
{\sqf} $\frm{t}$ is {\it bounded relative to} {\sqf} $\frm{h}$,
denoted $\frm{t} \relbd \frm{h}$, if
$\dom \frm{t} \supseteq \dom \frm{h}$ and there exist
$a,b > 0$ such that $|\frm{t}[\psi]| \le a {\bm 1}[\psi] + b |\frm{h}[\psi]|$
for every $\psi\in\dom\frm{h}$. The relation $\relbd$ is reflexive and transitive.

If $\frm{t}\relbd\frm{h}$ and $\frm{h}\relbd\frm{t}$, then
$\frm{t}$ and $\frm{h}$ are {\it equivalent}, denoted $\frm{t}\sim\frm{h}$.
Equivalent {\sqf}s have the same domain.

Sectoriality of $\frm{h}$ can be expressed as:
$\frm{h}^r$ is bounded below and \hbox{$\frm{h}^i \relbd \frm{h}^r$}.

$\relbd$ has a modest but useful calculus. For instance,
\begin{alignat}{3}
&    B\in\Lin(\sH)
& \;\Rightarrow\; &
B \relbd \frm{h},
    \nonumber \\
& c\in\Cmplx\setminus\{0\}
& \;\Rightarrow\; &
\frm{h} \sim c\frm{h},
    \nonumber \\
&  \frm{t} \relbd \frm{h}
& \;\Rightarrow\; &
\frm{h}+\frm{t} \relbd \frm{h},
    \nonumber \\
&  \frm{t}, \frm{h} \text{ sectorial }
& \;\Rightarrow\; &
\frm{h} \relbd \frm{h}+\frm{t}.
                   \nonumber
\end{alignat}

\item 
  \label{item:Cauchy}
A sequence 
$(\psi_n)$ in $\dom \frm{h}$ is {\it $\frm{h}$-Cauchy} if
\hbox{$(|\frm{h}| + {\bm 1})[\psi_n-\psi_m]\to 0$}.
It \hbox{\it $\frm{h}$-converges} to $\psi$ if
\hbox{$(|\frm{h}| + {\bm 1})[\psi_n-\psi]\to 0$}.
$\frm{h}$ is {\it closed} if all {$\frm{h}$-Cauchy} sequences 
$\frm{h}$-converge, {\it closable} if it has a closed extension. 

Note that $\frm{t}\relbd\frm{h}$ is equivalent to
every $\frm{h}$-Cauchy sequence is $\frm{t}$-Cauchy.

\end{enumerate}
\subsection{Completion and closure}
\label{sec:closure}

The notion of Cauchy-ness in item \ref{item:Cauchy} above is
common across an equivalence ($\sim$) class of {\sqf}s.
This is an important fact, as it points the way to a
``completion'' of an entire equivalence class on a common domain.
Therefore, we consider an equivalence class $\calC$ of {\sqf}s defined
on a dense subspace $\sK\subseteq \sH$, containing
a sectorical {\sqf} $\frm{h}$, and therefore a hermitian {\sqf} $\frm{h}^+ \ge {\bm 1}$.
The class of all forms on $\sK$ which are bounded relative to those in $\calC$
is denoted $\calC_\relbd$. 
The various sets of {\sqf}s involved here are related as
\begin{equation}
\sct{\calC} = \calC\cap \sct{\SF}(\sK)  \subset \calC \subset \calC_\relbd \subset \SF(\sK).
\end{equation}
The set $\sct{\calC}$, the sectorial forms among $\calC$, is a cone, while
$\calC_\relbd$ is a vector space. It will emerge that it has a natural Banach
space structure, up to norm-equivalence.

Two ${\mathcal C}$-Cauchy sequences $(x_n)$ and $(y_n)$ are equivalent
if \hbox{$(x_n-y_n)$} is $\calC$-Cauchy. This is written as $x \sim y$,
and the equivalence class of $(x_n)$ is denoted $x^\sim$.
Vectors in $\sK$ are identified with the classes of constant sequences.
The completion of ${\mathcal C}$ is constructed on the vector space
\begin{equation}
  \nonumber
\sKc \defeq\{\sim\text{-classes of }
{\mathcal C}\text{-Cauchy sequences in }\sK\},
\end{equation}
and $s$-forms in $\calC$ are extended to $\sKc$ according to
\begin{equation}
\label{eq:K-bar-forms}
\Dbraket{{x^\sim}}{\frm{t}}{{y^\sim}}  
\defeq \lim_{n\to\infty} \Dbraket{x_n}{\frm{t}}{y_n},
\end{equation}
as we now discuss.

The term {\it completion} suggests that we are dealing with the ordinary
completion of a relevant preHilbert space structure on $\sK$.
That is correct, and the inner product represented by any
$\frm{h}^+ \ge {\bm 1}$ in $\calC$ will do.
Let $\frm{h} \in \calC$ be sectorial, $\frm{h}^+$ as in
item \ref{item:sectorial} above, and
$(\sK,\frm{h}^+)$ be the preHilbert space structure consisting of the
space $\sK$ with inner product
$\inpr{{\phi}}{{\psi}}_{\frm{h}} \defeq \Dbraket{{\phi}}{\frm{h}^+}{{\psi}}$.
$\calC$-Cauchy is the same thing as $(\sK,\frm{h}^+)$-Cauchy in the usual sense,
and the usual Hilbert space completion of $(\sK,\frm{h}^+)$ can be viewed
as being carried on $\sKc$.
In order to see that {\sqf}s in $\calC$ can be extended to $\sKc$, we need
to know that they satisfy a Cauchy-Schwarz-like inequality.
\begin{lem}
  \label{lem:quasi-Cauchy-Schwarz}
Suppose ${\bm 1} \le \frm{h}^+$ and $\frm{t} \relbd \frm{h}$.
Then, there is some $M > 0$ such that for every $x,y\in\dom \frm{h}^+$,
\begin{equation}
\label{eq:pseudo-CS}
|\Dbraket{x}{\frm{t}}{y}|^2 \le M \frm{h}^+[x] \frm{h}^+[y]    
\end{equation}
\end{lem}
\begin{proof}
Only the case $\frm{t}$ hermitian, $|\frm{t}| \le \frm{h}^+$,
$\frm{t}[x,y]$ real, $\frm{h}^+[x] = \frm{h}^+[y] = 1$ need be checked, since
the general case follows by rescaling, multiplying $x$ by a phase $e^{i\theta}$,
and $|\frm{t}[x,y]| \le |\frm{t}^r[x,y]| + |\frm{t}^i[x,y]|$.
Here is the verification of the special case:
\begin{align}
  4|\frm{t}[x,y]|
  &=
  \frm{t}[x+y] - \frm{t}[x-y]  
    \nonumber \\
& \le  |\frm{t}[x+y]| + |\frm{t}[x-y]| 
    \nonumber \\
& \le  \frm{h}^+[x+y] + \frm{h}^+[x-y] = 4
    \nonumber 
  \end{align}
\end{proof}
This lemma asserts that every $\frm{t}\in\calC_\relbd$ is a bounded sesquilinear form on
the dense subspace $\sK$ of $(\sKc,\frm{h}^+)$, hence extends
by continuity to the full space so as to satisfy (\ref{eq:K-bar-forms}).
Each such extended {\sqf} is represented by a bounded operator on 
$(\sKc,\frm{h}^+)$; for instance, $\frm{h}^+$ itself is represented
by the identity. 

However, we also desire to identify $\sKc$ with a subspace of the ambient Hilbert
space $\sH$.
Certainly, the inclusion $\iota\colon {\sK} \hookrightarrow {\sH}$ extends by continuity
to a bounded
operator $\Arr{(\sKc,{\frm{h}^+})}{\tilde{\iota}}{\sH}$. The only question is whether it
is injective. It fails to be so only if there are two inequivalent $\calC$-Cauchy sequences
in $\sK$, which converge as sequences in $\sH$ to the same vector.
By linearity, only the case $x_n \to 0$ in $\sH$ need be considered:
$x \sim 0$ fails if and only if $\frm{t}[x_n] \not\to 0$, for any $\frm{t}\in\calC$.
The test may therefore be performed for any member of the class $\calC$.
If $\tilde{\iota}$ is injective, we simply identify $\sKc$ with
its image, and thereby obtain a closed {\sqf} in $\sH$ for every $\frm{t}$ in $\calC$.
In that case, $\calC$ is said to be closable.
Although it must be checked, only closable classes are of interest to us, so
closability is assumed henceforth.

\subsection{From {\sqf}s to operators}
\label{sec:ops-from-forms}

With $(\sKc,{\frm{h}^+})$ in the role of $\sHp$,
we obtain a Hilbert rigging as in Section~\ref{sec:Hilbert-rigging-abstract},
from which we now take over various notations.

Any bounded sesquilinear form $\frm{t}$ on $\sHp$,
(in particular, one in $\calC$) is represented by a unique operator
$\Opp{+}{+}{\frm{t}}\in\Lin(\sHp)$ satisfying
\begin{equation}
\inpr{\phi}{\Opp{+}{+}{\frm{t}}\psi}_{+}
= \Dbraket{\phi}{\frm{t}}{\psi}
\end{equation}
for all $\phi,\psi\in\sK$.
Using the unitary isomorphism $\Arr{\sHp}{J}{\sHm}$, we get another ``representation''
\hbox{$\Opp{+}{-}{\frm{t}} \defeq J \Opp{+}{+}{\frm{t}} \in\Lin(\sHp;\sHm)$}
of $\frm{t}$ satisfying
\begin{equation}
\inpr{\phi}{\Opp{+}{-}{\frm{t}}\psi}
= \Dbraket{\phi}{\frm{t}}{\psi}.
\end{equation}
[Recall that the inner product on $\sH$ extends to \hbox{$\sHm\times\sHp \cup \sHp\times\sHm$} as
in (\ref{eq:dual-with-respect-to-H}).]
The notation $\Opp{a}{b}{\frm{t}}$ indicates the domain and
range spaces in the subscript and superscript, respectively. It is unambiguous, but cumbersome.
Fortunately, it will not be needed much.
Restricting $\Opp{+}{-}{\frm{t}}$ to those $\psi$ such that 
$\Opp{+}{-}{\frm{t}}\psi$ is in $\sH$ yields yet a third operator,
$\Opp{0}{0}{\frm{t}} \in\Lin_0(\sH)$. The domain and range of this operator are
subspaces of $\sH$.
\subsection{Holomorphy of the $\Rmap$-map}
\label{sec:R-map}

The operator guise of $\frm{t}$ which is ultimately of most interest
is $\Opp{0}{0}{\phantom{\frm{t}}}$.
However, the $\Opp{+}{-}{\phantom{\frm{t}}}$ and $\Opp{+}{+}{\phantom{\frm{t}}}$ forms
have some especially nice properties, collectively:
\begin{thm}
  \label{thm:[]+-}
  the map $\frm{t}\mapsto \Opp{+}{-}{\frm{t}}$ is a bijection between
$\calC_{\relbd}$ and $\Lin(\sHp;\sHm)$.
The image $\Opp{+}{-}{\sct{\calC}}$ of $\sct{\calC}$ under this map
is an open subset of \hbox{$\Linv(\sHp;\sHm) - \Real_+$}.
\end{thm}
\begin{proof}[Proof of Thm. \ref{thm:[]+-}, part 1]
We already know from Section \ref{sec:closure} that
there is a natural bijection between $\calC_{\relbd}$ and $\Lin(\sHp)$.
By means of the unitary $J$, this is mapped into $\Lin(\sHp;\sHm)$.
\end{proof}
\begin{cnvntn}
  \label{cnvntn:C<-as-Banach-space}
From now on, we consider $\calC_\relbd$ to be equipped with this
Banach space structure --- up to norm-equivalence.
This structure is independent of the choice of $\frm{h}^+$ used to
construct $\sHp$, and therefore intrinsic.
\end{cnvntn}
The proof of the second part of Thm.~\ref{thm:[]+-} relies on the
following three Lemmas.
\begin{lem}
  \label{lem:adjoint-ops}
  If $\frm{t}\in\sct{\calC}$, then
  $\Opp{a}{b}{\frm{t}^*} = ( \Opp{a}{b}{\frm{t}} )^*$ 
for all choices of $a$ and $b$.
(N.B., the two $*$'s mean slightly different things.)
\end{lem}
\begin{proof}
  For $\Opp{+}{+}{\frm{t}}$ and $\Opp{+}{-}{\frm{t}}$,
  this is a simple matter of checking defintions.
  $\Opp{0}{0}{\frm{t}}$ involves some consideration of domains.
  $\psi\in\sHp$ is in $\dom \Opp{0}{0}{\frm{t}^*}$ iff
  \hbox{$\phi \mapsto \Dbraket{\psi}{\frm{t}}{\phi}$}
  extends to a bounded functional on $\sH$,
  whereas $\psi\in\sH$ is in $\dom (\Opp{0}{0}{\frm{t}})^*$ iff
  \hbox{$\phi \mapsto \inpr{\psi}{T \phi}$} does so.
  Hence
  $\Opp{0}{0}{\frm{t}^*} \subseteq ( \Opp{0}{0}{\frm{t}} )^*$
  is clear. To see the opposite, recognize that these are both
  closed sectorial operators, and without loss we may suppose that
  they are both {\em surjective}. 
\end{proof}

\begin{lem}
  \label{lem:sectorial-to-iso}
If $\frm{t}\in\sct{\calC}$ satisfies ${\bm 1}\le \frm{t}^r$, then
\hbox{$\Opp{+}{-}{\frm{t}}\in \Linv(\sHp;\sHm)$}.
\end{lem}
\begin{proof}
For notational simplicity,
set $T\defeq \Opp{+}{-}{\frm{t}}$.
Also, we may assume that $\frm{t}^r$ dominates $\|\cdot\|_+^2$ without loss, since some
multiple does so. 
\smallskip\newline
$\ker T=\{0\}$ and $\rng T$ closed:
$\|\psi\|_{+} \le \|{T}\psi\|_{-}$ follows from
\hbox{$\|\psi\|_+^2 \le | \Dbraket{\psi}{\frm{t}}{\psi} | 
= | \ilinpr{\psi}{{T}\psi}_{0}|
\le \|{T}\psi\|_{-}\|\psi\|_{+}$}.
\smallskip\newline
$\rng T$ dense:
$(\rng {T})^\perp = \ker {T}^*$ and $|\frm{t}^*| = |\frm{t}|$.
By Lemma~\ref{lem:adjoint-ops}, $\ker {T}^* = \{0\}$ follows just as
$\ker {T} = \{0\}$ above.
\smallskip\newline
$\rng T = \sHm$: $\rng {T}$ is both closed and dense in $\sHm$.
\end{proof}
\begin{lem}
  \label{lem:sectorial-usc}
  Suppose $\Sigma$ is an ample sector for $\frm{t}$.
  Then, $\Sigma$ is an ample sector for all $\frm{s}$ in some
  neighborhood of $\frm{t}$.
\end{lem}
\begin{proof}
  Without loss of generality, we may add a constant to $\frm{t}$ so that
  ${\bm 1}\le \frm{t}$, and choose the form used to turn $\sKc$ into
  a Hilbert space such that $\Opp{+}{+}{\frm{t}} = 1+iK$, with $K$ hermitian
  operator in $\Lin(\sHp)$.
  Then, with
  \hbox{$\Opp{+}{+}{\frm{s}} = (1+A) + i (K+B)$},
\begin{align}
  \left|  \frm{t}[\psi]  -  \frm{s}[\psi]    \right|
  &
  = \left| \inpr{\psi}{(A+iB)\psi}_+ \right|
  \nonumber \\ &
  \le (\|A\|+\|B\|)\|\psi\|_+^2
  \nonumber \\ &
  \le (\|A\|+\|B\|) \frm{t}[\psi]
  \nonumber 
\end{align}
\end{proof}

\begin{proof}[Proof of Thm.~\ref{thm:[]+-}, part 2]
If $\frm{t}\in\sct{\calC}$, then for some $m > 0$,
$\frm{t}+m{\bm 1}$ satisfies the hypotheses of
Lemma~\ref{lem:sectorial-to-iso}.
It follows that 
$\Opp{+}{-}{\frm{t}} \in \Linv(\sHp;\sHm) - \Real_+$.
It only remains to show that some neighborhood of 
$\Opp{+}{-}{\frm{t}}$ in $\Linv(\sHp;\sHm)$ corresponds to
sectorial forms. This, follows from Lemma~\ref{lem:sectorial-usc}.
\end{proof}
The pieces are now in place for a holomorphy-of-resolvent type result.
\begin{cnvntn}
  \label{cnvntn:Rmap}
  If $H\in\Lincl(\sH)$, then $\Rmap(\zeta,H)$ is the resolvent
  of $H$ at $\zeta$. This is thought of as a function of $\zeta$, in a context
  specified by $H$.
  We will overload this notation, writing
  $\Rmap(\zeta,{\frm{h}})$ for $\Rmap(\zeta,\Opp{0}{0}{\frm{h}})$,
  or in the case of an explicit parameterization,
  $\Rmap(\zeta,x)$ for $\Rmap(\zeta,H_x)$.
    In the latter two cases, we use the name $\Rmap$-map for $\Rmap$
    (even though that's redundant), rather than {\it resolvent}.
    The $\Rmap$-map has two arguments; the context is specified by a \RSF.
\end{cnvntn}
We show now that the $\Rmap$-map is holomorphic on 
\begin{equation}
\label{eq:Omega-def}  
\Omega \defeq
\setof{ (\zeta,\frm{h})\in\Cmplx\times\sct{\calC} }{ \zeta\in\res\Opp{0}{0}{\frm{h}} }.
\end{equation}
Since \hbox{$(\zeta,\frm{h}) \mapsto \Opp{+}{-}{\frm{h}}-\zeta \in \Lin(\sHp;\sHm)$}
is linear, this reduces to the question
(recall Convention~\ref{cnvntn:hats}) whether $\hat{T}\mapsto {T}^{-1}$ is holomorphic on
the subset of $\Linv(\sHp;\sHm)-\Cmplx$ where it is well-defined.
Prop.~\ref{prop:iso-to-closed} addressed the case of $\Linv(\sHp;\sHm)$, and
it is now a simple matter to extend it:
\begin{prop}
  \label{prop:final-piece}
  Given $\hat{T}\in\Linv(\sHp;\sHm) + \Lin(\sH)$.
  \newline\noindent\textnormal{(a)}
  ${T}\in\Lincl(\sH)$.
  \newline\noindent\textnormal{(b)}
  ${T}$ is injective iff $\hat{T}$ is injective.
  \newline\noindent\textnormal{(c)}
If \hbox{$\Arr{ \dom{T} }{{T}}{\sH}$} is bijective,
\hbox{$\hat{T}\in\Linv(\sHp;\sHm)$}.
\end{prop}
\begin{proof}
\noindent\textnormal{(a)}
This follows immediately from
Prop.~\ref{prop:iso-to-closed} and Lemma~\ref{lem:stability-of-closedness}.

\noindent\textnormal{(b)}
If $\hat{T}\phi=0$, then $\phi\in\dom{T}$.

\noindent\textnormal{(c)}
By assumption,
\hbox{$\hat{T}+B\in\Linv(\sHp;\sHm)$} for some \hbox{$B\in\Lin(\sH)$}.
Hence, given $\xi\in\sHm$, there is \hbox{$\phi\in\sHp$} such that
\hbox{$\xi = (\hat{T}+B)\phi = \hat{T}\phi + B\phi$}. But $B\phi\in\sH$, so
the equation $B\phi = \hat{T}\psi$ can be solved for $\psi\in\sHp$,
yielding $\xi = \hat{T}(\phi+\psi)$. That is, $\hat{T}$ is not only bounded,
but bijective as well, so $\hat{T}\in\Linv(\sHp;\sHm)$ (Open Mapping Theorem).
\end{proof}
Therefore, the supposed extension from
$\Linv(\sHp;\sHm)$ to $\Linv(\sHp;\sHm) - \Cmplx$ 
is illusory; all the operators we are interested in
here are {\em actually} already in the former set.
The following main result now follows immediately from the preceding work.
\begin{thm}
  \label{thm:resolvent-holo}
For $\frm{h}\in\sct{\calC}$, the closed operator $H = \Opp{0}{0}{\frm{h}}$ has an inverse
in $\Lin(\sH)$ iff  
$\hat{H} = \Opp{+}{-}{\frm{h}}$ has an inverse in $\Lin(\sHm;\sHp)$,
and $\Rmap$ is holomorphic on $\Omega$ \textnormal{[see (\ref{eq:Omega-def})]}.
\end{thm}
\subsection{Series expansion}
\label{sec:series}

We can reframe some of the main result in terms of the simplest ideas about series expansions.
Suppose $\hat{H}$ is in $Linv(\sHp;\sHm)$ and
$\hat{T}$ is in $\Lin(\sHp;\sHm)$.
Then, $\hat{H}^{-1}$ exists in $\Linv(\sHp;\sHm)$.
In case \hbox{$\|\hat{T}\|_{\Lin(\sHp;\sHm)} < (\|\hat{H}^{-1}\|_{\Lin(\sHm;\sHp)})^{-1}$},
both \hbox{$\hat{T} \hat{H}^{-1}\in\Lin(\sHm)$}
and \hbox{$ \hat{H}^{-1}\hat{T}\in\Lin(\sHp)$} are operators of norm less than one, and
\begin{align}
  (\hat{H}+\hat{T})^{-1}
&  = \sum_{n=0}^\infty (\hat{H}^{-1} \hat{T})^n \hat{H}^{-1}
                           \nonumber \\
  &  =  \hat{H}^{-1} \sum_{n=0}^\infty (\hat{T} \hat{H}^{-1})^n.
  \label{eq:quasi-resolvent-series}
\end{align}
Therefore, if $\hat{H} = \Opp{+}{-}{\frm{h}}$ and $\hat{T} = \Opp{+}{-}{\frm{t}}$,
we have a more or less explicit formula for $(\Opp{0}{0}{\frm{h}+\frm{t}})^{-1}$,
which we write $(H+T)^{-1}$ (recognizing that `$+$' here must be interpreted indirectly):
Merely sandwich the expansions in (\ref{eq:quasi-resolvent-series}) between $\iota_0$
and $\iota_+$.
This exhibits holomorphy in a very direct way. However, it does not itself show that
$H+T$ (or even $H$) is closed, nor does it show that invertibility of $H$ implies
invertibility of $\hat{H}$.

\subsection{Holomorphic families}
\label{sec:holomorphic-families}

We now return to the idea of parameterizing families of sectorial forms by
an open set in a Banach space.

\begin{defn}
  \label{def:holofam}
  Let $\sK$ be a dense subset of Hilbert space $\sH$, 
and $\calU$ a {\em connected} open subset of a Banach space.
The map $\Arr{\calU}{\frm{h}}{\SF(\sK)}$ is a
\newline
\noindent\textnormal{(a)} {\it [G-]holomorphic family}
in $\SF(\sK)$ parameterized over $\calU$ iff
\hbox{$x \mapsto \frm{h}_x[\psi] \Type{\calU}{\Cmplx}$}
is \hbox{[G-]holomorphic} for each $\psi\in \sK$.
\newline
\noindent\textnormal{(b)} {\it regular sectorial family}
in $\sct{\calC}$ parameterized over $\calU$ iff
$\Arr{\calU}{\frm{h}}{{\calC}_\relbd}$ is holomorphic
with range in $\sct{\calC}$, where $\calC$
is an equivalence class of closable {\sqf}s on $\sK$
and $\calC_\relbd$ has the Banach space structure of Convention~\ref{cnvntn:C<-as-Banach-space}.
\end{defn}
Various adjectives (``in $\SF(\sK)$/$\sct{\calC}$'', ``parameterized over $\calU$'')
may be omitted when context disambiguates.
\begin{rem}
By polarization, holomorphy of $\frm{h}$ immediately
implies that $\frm{h}_x[\phi,\psi]$ is holomorphic in $x$ for
every $\phi,\psi \in \sK$.
\end{rem}

If $\Arr{\calU}{\frm{h}}{\sct{\calC}}$ is a regular sectorial family, then
the composition of $\frm{h}$ with any holomorphic function on $\sct{\calC}$,
such as $\Rmap$ or (as shown in Section \ref{sec:free-energy}) $\Emap$,
is automatically holomorphic:
\begin{cor}
  \label{cor:Rmap-holo}
  If $\frm{h}$ defined on $\calU$ is a \RSF, then $\Rmap$ is
  holomorphic from its open domain in $\Cmplx\times \calU$ into $\Lin(\sH)$.
\end{cor}

On the other hand, the requirement to be merely a holomorphic family is weak
and easily checkable in application. Hence to get the abstract machinery
appropriately hooked up to specific parameterized families, the only real question
is when a holomorphic or G-holomorphic family is actually regular sectorial.
\begin{prop}
\label{prop:regular-sectorial}
A G-holomorphic family \hbox{$\Arr{\calU}{\frm{h}}{\sct{\SF}(\sK)}$}
is regular sectorial if any of the following criteria holds.
\begin{enumerate}
\item[\textnormal{(a)}]
$\frm{h}(\calU)$ consists of equivalent, closed {\sqf}s.
\item[\textnormal{(b)}]
{$\Arr{\calU}{\frm{h}}{\sct{\calC}\subset\calC_\relbd}$} is locally bounded
  for some closable class \hbox{$\calC\subseteq \SF(\sK)$}.
\item[\textnormal{(c)}]
$\frm{h}(\calU)$ consists of equivalent, closable {\sqf}s,
and for every $x$,
$\frm{h}_y$ is uniformly bounded with respect to $\frm{h}_x$
for $y$ in some neighborhood of $x$.
\end{enumerate}
\end{prop}
\begin{proof}
For (a), note that the assumption is that the forms $\frm{h}_x$ are already closed on $\sK$.
Hence, holomorphy amounts to weak-operator holomorphy on $\sHp$. Conclude with
Prop.~\ref{prop:WO-holo}(b).
For (b), appeal to Prop.~\ref{prop:WO-holo}(c).
Criterion (c) is a rephrasing of criterion (b) in light of the preceding theory.
\end{proof}


\subsection{Operator bounded families}
\label{sec:self-adjoint}

This subsection discusses an important kind of holomorphic family
constructed on the basis of a given lower-bounded self-adjoint operator $H$.
It is not used until Section \ref{sec:free-energy} and can safely be
skipped until then.

Take $H$ to be a lower-bounded self-adjoint operator in $\sH$,
and assume $1 \le H$, which can be arranged without loss by adding a constant.
Choose $\sK = \dom H$. 
On $\sK$, $H$ defines an \sqf\ $\frm{h}$ by
\begin{equation}
  \label{eq:h0}
\frm{h}[\psi] \defeq \inpr{\psi}{H\psi}.
\end{equation}
We denote the equivalence class of {\sqf}s to which $\frm{h}$ belongs
by $\calC(H)$, or simply by $\calC$ in this subsection, when there is no ambiguity.
Once it is known that $\calC$ is closable, the theory developed in this section
shows that \hbox{$\calC_\relbd \simeq \Lin(\sHp;\sHm)$}, where $\sHp$ is the completion
of $\dom H$ under the inner product
$\inpr{\phi}{\psi}_+ \defeq \inpr{\phi}{H\psi}$.
\begin{lem}
$\calC(H)$ is closable.  
\end{lem}
\begin{proof}
  $\sHp$ exists at least as the abstract completion of $\dom H$.
  Let $(\psi_n) \subset \dom H$ be an $\sHp$-Cauchy sequence,
  such that $\|\psi_n\|\to 0$, $\|\psi_n - \psi\|_+ \to 0$.
  It needs to be shown that $\psi= 0$.
  Taking the limit of
  \hbox{$\|\psi_n-\psi_m\|_+ =  \|\psi_n\|_+^2 + \|\psi_m\|_+^2 - 2 \re \inpr{\psi_n}{H\psi_m}$}
  as $n\to\infty$, yields \hbox{$0 = \|\psi\|_+^2 + \lim_{m\to\infty} \|\psi_m\|_+^2$},
  showing that $\psi=0$, as required.
\end{proof}

Define the {\em real} subspace $\sX^r(H)$ of $\SF(\dom H)$
to consist of hermitian {\sqf}s such that the norm
\begin{equation}
  \|\frm{t}\|_{H}
  = \sup
  \setof{  \frac{ |\Dbraket{\phi}{\frm{t}}{\psi}| }{\|{\phi}\| \|H{\psi}\|} }
  { 0 \neq \phi,\psi\in\dom H}.
\end{equation}
$\sX^r(H)$ corresponds precisely to the set of symmetric operators
on $\dom H$ which are operator bounded with respect to $H$.
Now, let
\begin{equation}
  \label{eq:OF}
\OF{H} \defeq \sX^r(H) \oplus i \sX^r(H)   
\end{equation}
be the complexification of $\sX^r(H)$, with the norm extended according to
\begin{equation}
  \|\frm{t}\|_H \defeq  \|\frm{t}^r\|_H +  \|\frm{t}^i\|_H.
\end{equation}
We aim to show that $\OF{H}$ is continuously embedded in $\calC(H)_\relbd$.
The following Lemma is the key step.
\begin{lem}
\label{lem:XH}
For $\frm{t}\in\sX^r(H)$,
\hbox{$|\frm{t}| < \|\frm{t}\|_{H}\,  \frm{h}$}.
\end{lem}
\begin{proof}
Assume \hbox{$\| \frm{t}\|_H = 1$}; the general case follows by homogeneity.
  
If for some $\psi$,
$|\frm{t}[\psi]| \ge \frm{h}[\psi]$, the numerical range of at least one of
$\frm{h} + \frm{t}$ and $\frm{h} - \frm{t}$ contains a negative number.
Therefore, it suffices to show that
\hbox{$0 \not\in \Num (\frm{h} + \frm{t})$}, and even,
by Prop.~\ref{prop:inf-Num-in-spec} below,
that $(-\infty,0]\in\res (H+T)$, where $T$ is the operator
on $\dom H$ induced by $\frm{t}$.
That will be the case if
$\|T\Rmap(x,H)\| < 1$ for $x \le 0$ (Lemma~\ref{lem:inverse-of-perturbed-closed-op}).
But this follows immediately from the definition of the norm $\|\cdot \|_H$:
\begin{equation}
  \nonumber
\|T\Rmap(x,H)\| < \|H\Rmap(x,H)\| \le 1.
\end{equation}
\end{proof}
The desired result follow immediately.
\begin{prop}
  \label{prop:X(H)}
Given lower-bounded self-adjoint operator $H$,
$\OF{H}$ is a Banach space continuously embedded in $\calC(H)_\relbd$.
Moreover, if $\|\frm{t}-\frm{h}\|_{H} < {1}$, then
$\oSec{0}{\frac{\pi}{4}}$ is a sector for $\frm{t}$.
\end{prop}

\subsection{Numerical range and spectrum}
\label{sec:Num-and-spec}

This subsection collections somewhat auxiliary results relating
the numerical ranges and spectra of operators.
In general, the relationship is subtle. 
Prop.~\ref{prop:resolvent-outside-Num} shows that the spectrum of a closed sectorial operator
is contained in the closure of its numerical range, but in general, the spectrum could
be much smaller: consider the matrix
$
\begin{pmatrix}
  0 & 1 \\
  0 & 0
\end{pmatrix}
$, with spectrum $\{0\}$ and numerical range a disk of radius $1/2$.

\begin{lem}
\label{lem:Num-and-spec-1}
  Let $S$ be a symmetric operator with numerical range in $[0,c]$ with $c < \infty$,
  and suppose that $(\phi_n)$ is a sequence of vectors such that
  $\inpr{\phi_n}{S\phi_n} \to 0$. Then, $S\phi_n\to 0$.
\end{lem}
\begin{proof}
  Assume for a contradiction that there is a subsequence $n(k)$
  such that $\|S\phi_{n(k)}\| > \epsilon > 0$.
  Without any loss, we may assume that the subsequence is the entire sequence.
  Hence, there exists a sequence of {\em unit} vectors $(\eta_n)$ such that
  \begin{align}
    \epsilon & \le |\inpr{\eta_n}{S\phi_n}|^2
\le \inpr{\eta_n}{S\eta_n} \inpr{\phi_n}{S\phi_n} 
        \nonumber \\
& \le c \inpr{\phi_n}{S\phi_n}  \to 0.
\nonumber
  \end{align}
  The second inequality here is Cauchy-Schwarz, and the contradiction
finishes the proof.
\end{proof}

\begin{prop}
\label{prop:inf-Num-in-spec}
For $T$ a symmetric operator, $\inf\Num T \in \spec T$.
\end{prop}
\begin{proof}
Assume that $\inf \Num T = 0$, since that can be arranged by adding a constant,
unless $\Num T$ is unbounded below, in which case the Lemma is vacuous anyway.
Then, there is a sequence $(\psi_n)$ of unit vectors in $\dom T$ such that
\begin{equation}
\label{eq:numerical-range-tending-to-zero}
\inpr{\psi_n}{T\psi_n} \to 0.
\end{equation}

Assume, for a contradiction, that \hbox{$0\in\res T$}, i.e.,
\hbox{$T^{-1}\in\Lin(\sH)$}. We will show this implies $\psi_n\to 0$.
Multiplying $T$ by a constant if necessary, we may assume $\|T^{-1}\|=1$.
Since
\begin{equation}
\inpr{\psi}{T\psi} = \inpr{T^{-1} T\psi}{T\psi} \in \|T\psi\| (\Num T^{-1}),
\end{equation}
non-negativity of $\Num T$ implies the same for $\Num T^{-1}$,
so that Lemma \ref{lem:Num-and-spec-1} will apply to $T^{-1}$.

Define \hbox{$\phi_n = T\psi_n$}.
Then,
\hbox{$\|\phi_n\| \ge \|\psi_n\| = 1$}
because $\|T^{-1}\| = 1$,
and (\ref{eq:numerical-range-tending-to-zero}) is rewritten as
\begin{equation}
  \nonumber
\inpr{\phi_n}{T^{-1}\phi_n} \to 0.  
\end{equation}
By Lemma \ref{lem:Num-and-spec-1} 
\hbox{$\psi_n = {T^{-1}\phi_n} \to 0$}. Contradiction.
\end{proof}

\begin{prop}
  \label{prop:resolvent-outside-Num}
  For an operator $T$,
  \newline
\noindent \textnormal{(a)}
  Each connected component of the open set
\hbox{$\Cmplx \setminus \cl{\Num T}$}
is either disjoint from $\res T$, or contained in it.
\newline
\noindent \textnormal{(b)}
In components contained in $\res T$, the resolvent is bounded as
\begin{equation}
  \label{eq:resolvent-bound}
\| \Rmap(\zeta,T) \| \le \frac{1}{\dist(\zeta,\Num T)}.
\end{equation}
\newline
\noindent \textnormal{(c)} If $T$ is {\em closed} sectorial, $\spec T \subseteq \cl\Num T$.
\end{prop}
\begin{proof}
For brevity, write $G$ for the {\em open} set \hbox{$\Cmplx\setminus\cl\Num T$}.

We first demonstrate the bound (\ref{eq:resolvent-bound})
for arbitrary \hbox{$\zeta\in G\cap\res T$},
and use that to show that both $G\cap \res T$ and $G\cap \spec T$ are open. 
That is equivalent to (a), and proves the remaining part of (b).

Thus, check for any unit vector $\psi\in\dom T$:
\begin{align}
  \|(T-\zeta)\psi\|
  & \ge |\inpr{\psi}{(T-\zeta)\psi}|
  \ge |\inpr{\psi}{T\psi} - \zeta|
    \nonumber \\
  & \ge \dist(\zeta,\Num T).
    \label{eq:T-zeta-closed-rng}
\end{align}
This establishes (\ref{eq:resolvent-bound}).

\noindent $G\cap\res T$ is open:
The open disk with center $\zeta$ and 
radius $\|\Rmap(\zeta,T)\|^{-1} \ge {\dist(\zeta,\Num T)}$
is contained in $\res T$.
(See Lemma~\ref{lem:inverse-of-perturbed-closed-op}.)

\noindent $G\cap \spec T$ is open:
For $\omega$ in $G\cap\spec T$, if there is $\zeta$ in $\res T$
with \hbox{$|\omega - \zeta| < \frac{1}{2}\dist(\omega,\Num T)$}, then  
\hbox{$ |\omega - \zeta| <  \dist(\zeta,\Num T) $}, contradicting the previous paragraph.

For part (c),
Since $\cl\Num T$ is convex,
it is geometrically more-or-less obvious that it is
either bounded, a closed sector, or bounded by two parallel lines.
The last is impossible since $T$ is sectorial, and $G$ has
exactly one component in either of the other two cases.
The conclusion follows from $\res T\neq\varnothing$ ($T$ is closed).
\end{proof}

\section{Magnetic Schr\"odinger forms}
\label{sec:QM}

The core theory of the previous section is inert on its own.
To use it, we need some interesting {\RSF}s, and some associated quantities
and objects which are holomorphic. The following two sections will take up
the latter issue.
This section is concerned with {\RSF}s of nonrelativistic Hamiltonians
which are parameterized by scalar and vector potential fields and a two-body
interaction. 
Though inteded to be more illustrative than exhaustive,
the results are nevertheless nontrivial.
See Section~\ref{sec:put-together} for the summary conclusion.

\subsection{Nonrelativistic $N$-particle systems}

We consider a system of $N$ identical particles moving in
three-dimensional euclidean space.
Hence, the ambient Hilbert space is $\sH \equiv L^2((\Real^3)^N)$.
As {\sqf}s, the Hamiltonians we wish to consider are sums
\begin{equation}
  \label{eq:total-hamiltonian-form}
\frm{h}_{{\bm A},u,v} = \kfrm{{\bm A}_0+{\bm A}}+\ufrm{u_0+u}+\vfrm{v_0+v},
\end{equation}
where 
\begin{equation}
  \label{eq:kA}
  {\frm{k}}_{\bm A}[\psi] =
\int_{\Real^{3N}} 
\sum_{\alpha=1}^{N} \left|[\nabla_\alpha - i{\bm A}(x_\alpha)] \psi\right|^2
\, dx,
\end{equation}
is kinetic energy with magnetic vector potential ${\bm A}$;
\begin{equation}
\label{eq:h1}
\ufrm{u}[\psi] = \int_{\Real^{3N}} \left[\sum_\alpha u(x_\alpha)\right] |\psi(x)|^2 \, dx
\end{equation}
is a one-body potential energy for scalar potential $u$; and
\begin{equation}
\label{eq:h2}
\vfrm{v}[\psi] = \int_{\Real^{3N}}
        \left[\frac{1}{2}\sum_{\alpha\neq\beta} v(x_\alpha-x_\beta)\right] |\psi(x)|^2 \, dx
\end{equation}
is a two-body interaction.
These are taken to be defined on the space $\sK \equiv C_c^\infty(\Real^{3N})$ of
compactly supported, infinitely differentiable functions.
In (\ref{eq:total-hamiltonian-form}), 
${\bm A}_0$, $u_0$ and $v_0$ are fixed background or unperturbed fields,
while ${\bm A}$, $u$ and $v$ are variable, drawn from appropriate Banach
spaces (to be determined) so that $\frm{h}_{{\bm A},u,v}$ is a regular sectorial family.

We do not say anything here about statistics because all the {\sqf}s/operators
to be considered are invariant under particle permutations; thus, one
can simply restrict attention to the subspace carrying the desired representation
of the permutation group. Taking spin explicitly into account is similarly unnecessary
since we are concerned with spin-independent Hamiltonians.

The uperturbed scalar and interaction potentials are taken to be locally integrable,
non-negative functions:
\begin{equation}
  \label{eq:u0-v0}
u_0, v_0 \in L_{\mathrm{loc}}(\Real^3)_+.
\end{equation}
An interesting and natural choice for the unperturbed scalar potential
$u_0$ is some kind of confining potential, e.g., $|x|^2$ or $|x|^4$.
It is actually somewhat artificial to consider an interaction potential which
could not be treated as a perturbation, since the Coulomb interaction $v(x) = |x|^{-1}$
can be.
Since we are not aiming for an exhaustive treatment, ${\bm A}_0$ is
dropped (or taken identically zero). Other choices complicate the analysis
considerably.

Local integrability of $u_0$ and $v_0$ ensures that $\sK$ is in the domains of
$\ufrm{u_0}$ and $\vfrm{v_0}$, while positivity then implies closability on $\sK$.
Indeed, $\ufrm{u_0}$ is closed on the space of functions square integrable
with respect to $[1 + \sum_\alpha u(x_\alpha)]\, dx$, and similarly for $\frm{v}_0$.
Closability of $\kfrm{0}$ was already considered in Section \ref{sec:kinetic-energy}.
Closability of the sum $\kfrm{0} + \ufrm{u_0} + \vfrm{v_0}$ is a new problem
(the pieces are not equivalent), which, however, is easily solved as follows
(see paragraph VI.1.6 of Kato\cite{Kato}):
if $\frm{t},\frm{s}$ are sectorial {\sqf}s closable on a common domain
$\sK$ (with closures $\overline{\frm{t}}$, $\overline{\frm{s}}$),
then $\frm{t}+\frm{s}$ is also closable on $\sK$.
This is true because a Cauchy sequence with respect to
$\frm{t}^+ + \frm{s}^+$ is Cauchy with respect to each of
$\frm{t}^+$ and $\frm{s}^+$ separately, hence has a limit in 
$\dom \overline{\frm{t}} \cap \dom \overline{\frm{s}}$.
By induction, this extends to any finite number of closable {\sqf}s.

In summary, the unperturbed form $\frm{h}_{{\bm 0},0,0}$ is sectorial and
closable on $\sK \equiv C_c^\infty((\Real^3)^n)$, being a member of an equivalence class
$\calC$ of forms on $\sK$. We are interested in pertubations 
$\frm{h}_{{\bm A},u,v}$ lying in $\sct{\calC}$, and which, moreover, vary holomorphically
with the parameters $({\bm A},u,v)$. Prop.~\ref{prop:regular-sectorial} will be the
main tool for demonstrating this.
Note that the identification of an unperturbed form is arbitrary. Any form in
$\sct{\calC}$ would do. Even on a practical level, this can be so to some extent.
For example, one may prefer $u_0(x) = R^4$ if $|x| \le R$, otherwise $|x|^4$ to the
``simpler'' $|x|^4$, and if there is a background magnetic field, different gauge choices
may recommend themselves.

\subsection{Scalar potential}
\label{sec:scalar_potential}

We begin the search for suitable perturbations with scalar potentials: bounded potentials, $L^{3/2}$
potentials, and proportional modifications $fu_0$ of the background potential.
Note that $\ufrm{u_0+u} = \ufrm{u_0} + \ufrm{u}$, so we work with 
$\ufrm{u}$ by itself, although $\ufrm{u_0}$ plays a role in determining which
$u$'s are acceptable.

The following new concept, an extreme sort of relative boundedness, is now going to
be very important.
\begin{defn}
For {\sqf}s $\frm{t}$ and $\frm{h}$:
  $\frm{t}$ is ({\it Kato}) {\it tiny} with respect to
  $\frm{h}$ iff, for any $b > 0$, 
  $|\frm{t}| \le a {\bm 1} + b |\frm{h}|$, for some $a\in\Real_+$.
  This is denoted $\frm{t} \ktiny \frm{h}$.
\end{defn}
Here are some simple yet useful properties of $\ktiny$.
\begin{enumerate}
\item $\frm{t}\sim {\bm 0} \;\Rightarrow\; \frm{t}\ktiny\frm{h}$ for any $\frm{h}$
\item
$\setof{\frm{t}}{\frm{t}\ktiny \frm{h}}$ is a vector space.
\item Given $\frm{t} \ktiny \frm{h}$:
  \begin{enumerate}
  \item $\frm{h}+\frm{t} \sim \frm{h}$
  \item $\frm{h}'\sim\frm{h}\;\Rightarrow\; \frm{t}\ktiny\frm{h}'$
  \item $\frm{h}$ sectorial $\Rightarrow\; \frm{h}+\frm{t}$ sectorial
  \item
    $\frm{h}$, $\frm{h}'$ sectorial $\Rightarrow \frm{t} \ktiny \frm{h}+\frm{h}'$
  \item $\frm{t}'\relbd \frm{t},\; \frm{h}\relbd \frm{h}'
    \;\Rightarrow \; \frm{t}' \ktiny \frm{h}'$
\end{enumerate}
\end{enumerate}

The main point here is that we can accumulate tiny perturbations indefinitely
without danger of moving out of $\sct{\calC}$. Because of item 3(b), it makes
sense to write $\frm{t}\ktiny \calC$. But, beware: a {\em set} of forms tiny with respect
to $\calC$ might still be unbounded in $\calC_\relbd$.

\subsubsection{Bounded potentials}

These are {\em complex} functions, $u\in L^\infty(\Real^3)$, even though ultimately
we are (probably) only interested in the real subspace $L^\infty(\Real^3;\Real)$.
This expansion is for the sake of holomorphy, as usual;
we need to work in the complex space to use the theory of Section \ref{sec:families}.

This simple case is a good illustration of the basic method:
check that the perturbation does not move $\frm{h}_{{\bm A},u,v}$ out of $\sct{\calC}$;
check G-holomorphy; check local boundedness.

Now,
\begin{equation}
  \label{eq:bounded-potl-trivial}
  |\ufrm{u}[\psi]| \le \|u\|_{L^\infty(\Real^3)} {\bm 1}[\psi].
\end{equation}
This immediately demonstrates local boundedness of $\frm{h}_{{\bm 0},u,0}$
due to the factor $\|u\|_{L^\infty(\Real^3)}$, as well as that
\hbox{$\ufrm{u} \sim 0 \ktiny \frm{h}_{{\bm 0},0,0}$}. 
G-holomorphy of $\ufrm{u}[\psi]$ in $u$ is trivial because it is {\em linear}.
In complete detail:
\begin{equation}
  \nonumber
  \ufrm{u+zu'}[\psi] = \ufrm{u}[\psi] +  z \ufrm{u'}[\psi],
\end{equation}
so the issue reduces to holomorphy of the right-hand side in $z$,
which only requires that $\ufrm{u}[\psi]$ and $\ufrm{u'}[\psi]$ be well-defined.
N.B. This argument has nothing to do with the topology of $L^\infty$ and will
hold for any vector space for which $\ufrm{u}\in\SF(\sK)$.
Conclusion: $L^\infty(\Real^3) \ni u \mapsto \frm{h}_{{\bm 0},u,0}$ is
a \RSF.

\subsubsection{Unbounded potentials}

Now we move on to a space of unbounded potentials, namely,
$u\in L^{3/2}(\Real^3)$.
The following Lemma provides a bound playing the same r\^ole as
(\ref{eq:bounded-potl-trivial}).
The Sobolev inequality used can be found in books on
Sobolev spaces\cite{Adams}, partial differential equations\cite{Taylor1}
and general analysis\cite{Lieb+Loss}.
\begin{lem}
  \label{lem:L3/2-bound}
  If $u\in L^{3/2}(\Real^3)$, then
  \begin{equation}
    \label{eq:L3/2}
|\ufrm{u}| \le c'' \|u\|_{L^{3/2}(\Real^3)} \frm{k}_0.  
\end{equation}
\end{lem}
\begin{proof}
For fixed $y \equiv (x_2,\ldots,x_N)$,  the H\"{o}lder inequality gives
\begin{align}
    \nonumber
 \int u(x_1) |\psi(x_1,y)|^2 \, dx_1
& \le c \|u\|_{L^{3/2}} \left\{ \int |\psi(x_1,y)|^6 \, dx \right\}^{1/3}
   \nonumber \\
& = c \|u\|_{L^{3/2}} \| \psi(\cdot ,y) \|_{ L^{3/2}(\Real^3) }^2
   \nonumber 
  \end{align}
For the integral here, use the Sobolev inequality 
\begin{equation}
  \|f\|_{L^q(\Real^d)} \le c' \|f\|_{W_k^p(\Real^d)},
  \quad p\le q \le \frac{pd}{d-kp}
\end{equation}
with the values $d=3$, $k=1$, $p=2$, $q=6$ to obtain
\begin{equation}
    \nonumber
    \int \| \psi(\cdot ,y) \|_{ L^{3/2}(\Real^3) }^2 \, dy
    \le c' 
    \int |\nabla_1\psi(x_1,y) |^2 \, dy.
\end{equation}
Adding up the inequalities with each of $x_2,\ldots,x_N$ in place of $x_1$ yields
\begin{equation}
\nonumber
\ufrm{u}[\psi] \le c c' \|u\|_{L^{3/2}} \kfrm{0}[\psi].
\end{equation}
\end{proof}
This demonstrates local boundedness of 
\hbox{$L^{3/2}(\Real^3)\ni u \mapsto \frm{h}_{{\bm 0},u,0}$}, but would allow us to
conclude that \hbox{$\frm{h}_{{\bm 0},u,0}\in\sct{\calC}$} only for $\|u\|_{L^{3/2}}$
sufficiently small (depending on $c''$). Fortunately, it can be improved
by using density of $L^\infty$ in $L^{3/2}$.
\begin{lem}
  \label{lem:L3/2-tiny}
For $u\in L^{3/2}(\Real^3)$, \hbox{$\ufrm{u} \ktiny \frm{k}_0$}.
\end{lem}
\begin{proof}
Split $u$ as $u = u' + u''$, with
$u'\in L^\infty(\Real^3)$ and $u'' \in L^{3/2}(\Real^3)$.
$\|u''\|_{L^{3/2}(\Real^3)}$ can be made as small as desired by
choosing $u'$ appropriately (e.g. $u' = u \, 1[|u|\le M]$ for large $M$).
\end{proof}
Conclusion: $L^{3/2}(\Real^3) \ni u \mapsto \frm{h}_{{\bm 0},u,0}$ is
a \RSF.

\begin{rem}
This result is very important to Lieb's framework\cite{Lieb83} for DFT.
\end{rem}

\subsubsection{Modulating the confining potential}

The final kind of scalar potential to be considered is modulation of
the background (confining) potential: $L^\infty(\Real^3) \ni f \mapsto \ufrm{fu_0}$.
Evidently,
\begin{equation}
|\ufrm{fu_0}|  \le \|f\|_{L^\infty(\Real^3)} \ufrm{u_0}.
\end{equation}
Local boundedness is thus secure, but $\frm{h}_{{\bm 0},fu_0,0}$ will
generally fail to be sectorial if $u_0$ is anything like what we have in mind.
Thus, we need to restrict $f$ to the open unit ball $B(L^\infty(\Real^3))$.
With that restriction, another \RSF\ is obtained.

\subsubsection{Removing redundancy}

Combine the preceding three kinds of scalar potential perturbation yields
a holomorphic map
\begin{equation}
  \nonumber
L^\infty(\Real^3) \oplus  L^{3/2}(\Real^3) \oplus  L^\infty(\Real^3) \to \calC_\relbd
\end{equation}
given by
\begin{equation}
\label{eq:total-scalar-potl-map}
  (u',u'',f) \mapsto
  \ufrm{u'}+ \ufrm{u''}+  \ufrm{fu_0} =  \ufrm{u'+u''+fu_0}.
\end{equation}
However, this should be restricted to the open set
\begin{equation}
 \label{eq:U-for-u}
\calU \defeq \setof{(u',u'',f) \in L^\infty \oplus  L^{3/2} \oplus  L^\infty }{\|f\| < 1}
\end{equation}
to ensure that $\frm{h}_{{\bm 0},u'+u''+fu_0,0}$ is in $\sct{\calC}$.
Thus, we have a \RSF\ in $\sct{\calC}$ parameterized over $\calU$ above. 

However, this is not entirely satisfactory because there is redundancy:
many distinct triples $(u',u'',f)$ may give the same total potential $u' + u'' + fu_0$.
To cure this infelicity, we pass to a quotient.
Recall that the quotient $\sX/{\mathscr M}$ of a Banach space $\sX$ by a
closed subspace ${\mathscr M}$ is a Banach space with norm
\begin{equation}
\nonumber  
\|\pi x\|_{\sX/{\mathscr M}} \defeq \inf \setof{\|x+m\|_{\sX}}{m\in{\mathscr M}},
\end{equation}
where $\Arr{\sX}{\pi}{\sX/{\mathscr M}}$ is the canonical projection.
A continuous linear map $\Arr{\sX}{f}{\sY}$ naturally induces a linear map on the
quotient $\sX/\ker f$, eliminating directions along which $f$ is constant.

This simple picture is complicated in situations which interest us for two
reasons. $f$ is not sure to be either linear or defined on the entire
space $\sX$, hence a slightly generalized notion of kernel is
needed, and taking a quotient by a subspace is not an immediately sensible
thing to do. 
\begin{lem}
Given ${\calU \subseteq \sX}$ open and {\em convex}, and
\hbox{$\Arrtop{\calU}{f}{\sY}$}, holomorphic, let
\begin{equation}
  \nonumber
{\mathscr M} = \cap_{x\in\calU} \ker Df(x).  
\end{equation}
Then, $f$ has a unique holomorphic extension to \hbox{$\calU + {\mathscr M}$},
given by \hbox{$f(x+m) = f(x)$} for \hbox{$m\in{\mathscr M}$}.
In turn, a holomorphic map
\hbox{$\Arr{(\calU+{\mathscr M})/{\mathscr M}}{\tilde{f}}{\sY}$ }
is induced on the quotient,
given by $\tilde{f}(\pi x) = f(x)$.
\end{lem}
\begin{proof}
  First, note that $\ker Df(x)$ is a closed subspace of $\sX$ for each $x\in\calU$,
  so ${\mathscr M}$ is indeed a closed subspace. To see that the asserted extension is well-defined,
  suppose that $y = x+m = x'+m'$, for $x,x'\in\calU$, $m,m'\in{\mathscr M}$.
  Denote the affine (two-$\Cmplx$-dimensional) subspace containing $x,x',y$ by $A$,
  and consider the restriction of $f$ to $A\cap\calU$, which is convex. The restriction of
  $Df$ is everywhere zero, hence $f$ is constant on $A\cap\calU$, i.e., $f(x)=f(x')$ and the
  extended $f$ is well-defined. That the extension is holomorphic follows immediately
  from $Df(x+m) = Df(x)$, and unicity from $\calU+{\mathscr M}$ being connected and $f$ given
  on an open set, namely $\calU$.

  Therefore, $\tilde{f}$ is well-defined on $\calU/{\mathscr M}$ according to the given formula,
  and it remains only to show that it is holomorphic. As usual, we use the equivalence
  with G-holomorphy plus local boundedness (Thm. \ref{thm:GTHZ}).
  For G-holomorphy, note that $\tilde{f}(\pi x+\zeta \pi y) = f(x+\zeta y)$,
  so the question reduces to G-holomorphy of $f$ itself.
  For local boundedness, note that $\|\pi x - \tilde{y}\| < \epsilon$
  implies that $x$ is within distance $\epsilon$ of $\pi^{-1} \tilde{y}$.
\end{proof}

To apply this, one only need check that $\calU$ in (\ref{eq:U-for-u}) is convex,
which is immediate.
So, define \hbox{$L^\infty(\Real^3)+L^{3/2}(\Real^3)+u_0 L^\infty(\Real^3)$}
to be the space of functions (equivalence classes under a.e. equality)
$u$ such that
\begin{equation}
  \nonumber
\inf\setof{\|u'\|_{L^\infty} +\|u''\|_{L^{3/2}} +\|f\|_{L^{\infty}}}
{u=u'+u''+fu_0}
\end{equation}
is finite. This is a norm $\|u\|$ making $L^\infty+L^{3/2}+u_0 L^\infty$ a Banach space,
and the subset $\calU_u$ consisting of $u$ with
{\em some} decomposition obeying the constraint $\|f\|_{L^\infty} < 1$ is open.
Conclusion: the map $u \mapsto \frm{h}_{{\bm 0},u,0}$
is a \RSF\ parameterized over $\calU_u$.

\subsection{Interaction}
\label{sec:interaction}

The message of this subsection is that two-body interactions can be
treated in much the same way as one-body potentials, an observation that
goes back centuries. Indeed, instead of coordinatizing configuration space
$(\Real^3)^N$ with $x_1,x_2,\ldots,x_N$, we may use
$\tfrac{x_2-x_1}{\sqrt{2}},\tfrac{x_2+x_1}{\sqrt{2}},x_3,\ldots,x_N$,
and thereby control an interaction between particles $1$ and $2$ by the
kinetic energy just as an external potential for particle $1$.
As long as we use only Kato tiny perturbations, as was done in
Section~\ref{sec:scalar_potential}, then, owing to property 2 of
Section~\ref{sec:scalar_potential} it is not possible that each
perturbation alone is controllable, while the combination is not.
We have, for example, an \RSF\ of pair interactions 
parameterized over $\calU_{v} = L^\infty(\Real^3)+L^{3/2}(\Real^3)$.

\subsection{Vector potential}
\label{sec:vector_potential}

For out purposes, the form of ${\frm{k}}_{\bm A}$ given in (\ref{eq:kA}) is not good
for complex vector potentials.
In order that ${\frm{k}}_{\bm A}$ be holomorphic in ${\bm A}$, it should {\em not}
appear complex-conjugated.
The correct definition is
\begin{align}
\label{eq:complex-A}  
{\frm{k}}_{\bm A}[\psi]
&=
\sum_{\alpha=1}^{N}
\int_{\Real^{3N}} 
(\nabla_\alpha + i{\bm A}(x_\alpha)) \overline{\psi}\cdot
(\nabla_\alpha - i{\bm A}(x_\alpha)) \psi
            \, dx
\nonumber \\                          
&=
\inpr{ ({\nabla} - i{\overline{\bm A}})\psi }
{ ({\nabla} - i{\bm A})\psi }
\end{align}
We take a somewhat different approach with this than for scalar potentials.
$\nabla_\alpha$
is a bounded operator from $W_1^2(\Real^{3N})$ into $\vec{L}^2(\Real^{3N})$
(we use an over-arrow to indicate ordinary, complex, three-dimensional vectors).
The integral in (\ref{eq:complex-A}) will be a legitimate $L^2$ inner product if
multiplication by ${\bm A}$ (or $\overline{\bm A}$) has the same property.
This is very natural train of thought, but before pursuing it, we consider
bounded vector potentials.

\subsubsection{bounded \texorpdfstring{${\bm A}$}{vector potential}}
\label{sec:bounded-A}

\begin{lem}
If ${\bm A}$ is bounded, then $\kfrm{\bm A}$ is a tiny perturbation of $\frm{k}_0$.  
\end{lem}
\begin{proof}
For an arbitrary $\psi\in W_1^2(\Real^{3N})$,
\begin{align}
\Big|\kfrm{\bm A}[\psi] - \kfrm{0}[\psi]\Big| 
 = & \Big|\inpr{ ({\nabla} - i{\overline{\bm A}})\psi }
  { ({\nabla} - i{\bm A})\psi } - \|\nabla\psi\|^2\Big|
      \nonumber
 \\
\le
& 
      \|{\bm A}\|_{L^\infty}^2   \|\psi\|^2
      +2 \|{\bm A}\|_{L^\infty}  \|\nabla\psi\| \|\psi\|.
  \label{eq:k-A-tiny-pert}
\end{align}
Control the final term with the inequality
\begin{equation}
  \nonumber
  2 \|\nabla\psi\| \|\psi\| \le
\epsilon \|\nabla\psi\|^2
+\frac{1}{\epsilon}\|\psi\|^2, \quad \epsilon > 0.
\end{equation}
Since $\epsilon$ can be taken as small as desired here,
\begin{equation}
\kfrm{\bm A} - \kfrm{0} \ktiny \kfrm{0}.
\end{equation}
\end{proof}

Just as G-holomorphy of $\ufrm{u}$ followed from holomorphy of
$\Cmplx\ni z \mapsto z$, G-holomorphy of $\kfrm{\bm A}$ follows from
holomorphy of $z \mapsto z^2$.
Local boundedness of $\kfrm{\bm A}[\psi]$ as a function of
\hbox{${\bm A} \in \vec{L}^\infty(\Real^3)$} follows from an
estimate like that in (\ref{eq:k-A-tiny-pert}).
Thus,
$\vec{L}^\infty(\Real^3)\ni {\bm A} \mapsto \kfrm{\bm A}$ is a \RSF.

\subsubsection{Sobolev multipliers}

Now we return to the idea mentioned at the beginning of this section.
Multiplication of elements of $W_1^2(\Real^d)$
by a fixed function $f$ is a linear operation.
If it is actually a {\em bounded} linear operator into $L^2(\Real^d)$,
then $f$ is a member of the space
\hbox{$M(W_1^2(\Real^{d})\to L^2(\Real^{d}))$} of
Sobolev multipliers\cite{Mazya-85-book,Mazya-09-book}.
This space is nontrivial (it contains $L^\infty$) and is a Banach space
with the norm it inherits from $\Lin(W_1^2(\Real^{d});L^2(\Real^{d}))$:
\begin{equation}
\|f\|_{M(W_1^2\to L^2)} \defeq \sup\setof{\|f\psi\|_{L^2}}{\|\psi\|_{W_1^2} = 1}
\end{equation}

Therefore, we consider
\hbox{${\bm A} \in \vec{M}(W_1^2(\Real^{3});L^2(\Real^{3}))$}.
One needs to check that this lifts from 3-dimensional to $3N$-dimensional
space properly, but that is simple: abbreviating $y\equiv (x_2,\ldots,x_N)$,
\begin{align}
  \int |{\bm A}(x_1)\psi(x_1,y)|^2 & \, dx_1
\nonumber
\\ &  \le \|{\bm A}\|_{\vec{M}(W_1^2\to L^2)}
 \int |\nabla_1 \psi(x_1,y)|^2 \, dx_1.
  \nonumber
\end{align}
Integration over $y$ shows that the norm is independent of $N$.

G-holomorphy has nothing to do with the topology of the space over which
${\bm A}$ ranges, so it follows for
$\vec{M}(W_1^2(\Real^{3N});L^2(\Real^{3N}))$
just as for bounded vector potentials.
Local boundedness follows from a calculation much like
(\ref{eq:k-A-tiny-pert}):
\begin{align}
\Big|\kfrm{{\bm A}+{\bm a}}[\psi] - \kfrm{\bm A}[\psi]\Big| 
\le
& \| ({\nabla} - i{\bm A})\psi \|\| {\bm a}\psi \|
\nonumber \\
& + \| ({\nabla} - i\overline{\bm A})\psi \| \| {\bm a}\psi \|
+ \| {\bm a}\psi \|^2.
\nonumber 
\end{align}

This establishes that,
for \hbox{${\bm A}\in\vec{M}(W_1^2(\Real^{3N});L^2(\Real^{3N}))$},
$\kfrm{\bm A} \relbd \kfrm{0}$.
However, the opposite, $\kfrm{0} \relbd \kfrm{\bm A}$,
is problematic in general, although it does hold if
$\|{\bm A}\|_{M(W_1^2 \to L^2)} < 1$. 
 The situation looks at first like what we faced
with \hbox{$u\in L^{3/2}$} for $\ufrm{u}$. However, $L^\infty$ is not dense in
$M(W_1^2 \to L^2)$. The norm is an operator norm and we face the
familiar problem that strong convergence does not imply norm convergence.
Thus, we settle for what is clear,
$\kfrm{0} \sim \kfrm{\bm A}$ for ${\bm A}$ in the unit ball
$B(M(W_1^2 \to L^2))$. 

On one level the preceding is entirely satisfactory. The Sobolev-multiplier
norm is natural. However, one might prefer something more familiar and
easier to work with, such as given in the following Lemma.
\begin{lem}
\label{lem:vec-L3-tiny}
For ${\bm A} \in \vec{L}^{3}(\Real^3)$, \hbox{$\frm{k}_{\bm A} \sim \frm{k}_0$}.
\end{lem}
\begin{proof}
Use a H\"older inequality and the Sobolev inequality
cited in Lemma~\ref{lem:L3/2-bound} to obtain
\begin{equation}
  \label{eq:Holder-Sobolev}
  \|{\bm A}\psi\|_{L^2} \le \|{\bm A}\|_{L^3} \|\psi\|_{L^6}
  \le c \|{\bm A}\|_{L^3} \|\psi\|_{W_{1}^2}.  
\end{equation}
Again, just as in Lemma~\ref{lem:L3/2-tiny}, a bounded vector field can be
subtracted from ${\bm A}$ so that the $L^3$ norm of the residual is as small
as desired.
\end{proof}

\subsubsection{Removing redundancy again}

As for scalar potentials, there is also redundancy here, since
$\vec{L}^\infty$ intersects $\vec{L}^3$, and is contained in 
${\vec{M}(W_1^2\to L^2)}$. It can be solved in exactly the same way to
obtain a \RSF\ parameterized over an open set $\calU_{\bm A}$ in 
$\vec{L}^\infty + {\vec{M}(W_1^2\to L^2)}$ or all of $\vec{L}^\infty + \vec{L}^3$.

\subsection{Putting it all together}
\label{sec:put-together}

Here is the summary of preceding investigation.
With lower-bounded locally integrable background potential and interaction (\ref{eq:u0-v0})
and no background vector potential, $\frm{h}_{\bm A,u,v}$ (\ref{eq:total-hamiltonian-form})
is a \RSF\ on {\em all} of
\hbox{$(\vec{L^3} + \vec{L}^\infty)\times (L^{3/2} + L^\infty) \times  (L^{3/2} + L^\infty)$}.
Alternatively, the $L^3$ summand for ${\bm A}$ can be replaced by $M(W_1^2(\Real^{3}))$
and summands $u_0L^\infty(\Real^3)$ and $v_0L^\infty(\Real^3)$ added to the
potential and interaction factors with restriction to an open neighborhood $\calU$
of the origin. The condition to be in $\calU$ does not factorize.


\section{Low-energy Hamiltonians \& eigenstate properties}
\label{sec:eigenvalues}

The previous section was concerned with one component of application,
namely the construction of {\RSF}s useful for nonrelativistic quantum
mechanics. This section and the next tackle the question:
given an \RSF\ $\frm{h}$ defined on $\calU$, what interesting
functions/quantities are holomorphic? To a considerable extent, this
can be fruitfully discussed without reference to any concrete \RSF.
This section uses Riesz-Dunford-Taylor integral methods to discuss
``low-energy Hamiltonians'' in case there is a gap in the spectrum,
i.e., a curve $\Gamma$ in the resolvent set of $H_x$ running top-to-bottom in $\Cmplx$
(recall, we deal in ``Hamiltonians'' which are sectorial but not necessarily
self-adjoint). The part of the spectrum to the left of $\Gamma$ then
corresponds to a bounded Hamiltonian which is holomorphic on some neighborhood
of $x$. 
Properties of nondegenerate eigenstates associated with isolated
eigenvalues are considered in section \ref{sec:rank-1}.
The eigenvalue itself and expectations of all ordinary observables, as well
as of generalized observables such as charge-density and current-density
(when they make sense) are holomorphic.
Some of the material here, primarily Section \ref{sec:Riesz-Dunford-Taylor}
and Prop.~\ref{prop:automatic-holo-Lp} are appealed to in section \ref{sec:free-energy}.

\subsection{Riesz-Dunford-Taylor integrals}
\label{sec:Riesz-Dunford-Taylor}

Recall that one of the main conclusions of Section \ref{sec:families}
was holomorphy of the map $(\zeta,x) \mapsto \Rmap(\zeta,H_x)$.
As a function of the single complex variable $\zeta$, it is natural to
integrate this around contours. The Riesz-Dunford-Taylor
calculus constructs a holomorphic function $f(A)$ of an arbitrary {\em bounded}
operator $A$ by integrating $f(\zeta)\Rmap(\zeta,A)$ around a contour
encircling the entire spectrum $\spec A$, where $f$ is an ordinary holomorphic
function.
Some basic references for this technology are \S III.6 of Kato\cite{Kato}, 
Chap. 6 of Hislop \& Sigal\cite{Hislop+Sigal},
or \S 3.3 of Kadison \& Ringrose\cite{Kadison+Ringrose-I}.
Since we deal with unbounded operators, we cannot do that, but
the idea can be modified for some interesting purposes. 

The first basic idea is that,
if $H$ is a closed operator, $E$ is an isolated eigenvalue,
and ${\Gamma}$ is a simple anticlockwise closed
contour in $\res H$, surrounding $E$ but no other part of $\spec H$,
then
\begin{equation}
\label{eq:P(u)}
P(\Gamma) = -\oint_{\Gamma} \Rmap(\zeta,H) \frac{d\zeta}{2\pi i}.
\end{equation}
is a projection onto the corresponding eigenspace.
For normal (in particular, self-adjoint) operators this is straighforward
as the relevant part of the resolvent looks like $(E-\zeta)^{-1} P$, where
$P$ is an {\em orthogonal} projector onto the eigenspace, so that $P(\Gamma) = P$.
The restriction of a non-normal operator to an eigenspace is generally not
simply a multiple of the identity if the algebraic multiplicity exceeds one.
Consequently, the resolvent generally has higher-order poles.
If the eigenvalue is {\em nondegenerate} however, that cannot happen and
its value can be extracted as
\begin{equation}
\label{eq:ground-energy-integral}
E = -\Tr \oint_{\Gamma} \zeta \Rmap(\zeta, H) \frac{d\zeta}{2\pi i}.
\end{equation}

We can profitably generalize somewhat.
First, we have the basic result
\begin{prop}
\label{prop:Riesz-Dunford-Taylor}
Given: $A\in\Lincl(\sX)$ and $\Gamma$ a simple anticlockwise contour
    in $\res A$, surrounding the part $\sigma$ of $\spec A$.
Then,
\newline\noindent \textnormal{(a)}
\begin{equation}
  \label{eq:P(G)}
P(\Gamma) \defeq -\oint_\Gamma \Rmap(\zeta,A) \frac{d\zeta}{2\pi i}      
\end{equation}
is a projection with $\rng P(\Gamma) \subseteq \dom A$.
\newline\noindent \textnormal{(b)}
\begin{equation}
  \label{eq:A(G)}
A(\Gamma) \defeq -\oint_\Gamma \zeta \Rmap(\zeta,A) \frac{d\zeta}{2\pi i}      
\end{equation}
satisfies $A(\Gamma) = A P(\Gamma) = P(\Gamma) A P(\Gamma)$ (hence maps $\rng P(\Gamma)$ into
itself), annihilates $\ker P(\Gamma)$, 
and its spectrum as an operator on $\rng P(\Gamma)$ is $\sigma$.
\newline\noindent \textnormal{(c)}
More generally, for open $U$ containing $\Gamma$ and the region it surrounds,
and $f$ an ordinary holomorphic function on $U$,
\begin{equation}
  \label{eq:f(A|G)}
f(A|\Gamma) \defeq -\oint_\Gamma f(\zeta) \Rmap(\zeta,A) \frac{d\zeta}{2\pi i}      
\end{equation}
maps $\rng P(\Gamma)$ into itself and annihilates $\ker P(\Gamma)$.
$P(\Gamma) = 1(A|\Gamma)$ and $A(\Gamma) = \Id(A|\Gamma)$ are special cases.
$f\mapsto f(A|\Gamma)$ is a Banach algebra morphism from the space of holomorphic
functions on $U$ (with uniform norm) into $\Lin(\rng P(\Gamma))$.
\end{prop}
\begin{proof}
  For the parts concerning $P(\Gamma)$ and $A(\Gamma)$,
  see Hislop \& Sigal\cite{Hislop+Sigal}, Prop.~6.9. For the
  Banach algebra aspects, see Kadison and Ringrose.
\end{proof}

We are not nearly so interested in varying $f$ in (\ref{eq:f(A|G)}), however,
as in varying $A$ for a few simple cases of $f$, principally $1$ and $\Id$.
\begin{thm}
  \label{prop:f(y|G)-holo}
Given \RSF\ $\frm{h}$ and
simple closed contour $\Gamma \subset \res H_x$,
there is a neighborhood $\calW$ of $x$ such that
\hbox{$y\in\calW\;\Rightarrow\; \Gamma\subset \res H_y$} and
for each $f$ holomorphic on and inside $\Gamma$,
\hbox{$y \mapsto f(H_y|\Gamma) \colon \calW \to\Lin(\sH)$} is holomorphic.
\end{thm}
\begin{proof}
By compactness of $\Gamma$ and holomorphy of
  \hbox{$(\zeta,y)\mapsto \Rmap(\zeta,H_y)$}.  
\end{proof}

\subsection{Low-energy Hamiltonians}

No contour can be drawn around the entire spectrum of an unbounded operator $H_x$.
However, since $H_x$ is bounded below, it might be possible to surround the part of $\spec H_x$
in some left-half-plane, if there is a gap.
That such a contour will continue to surround the ``low energy'' part of the spectrum
when $x$ is perturbed is not immediately evident: Each $H_y$ is bounded below,
but is it possible that $\spec H_y$ has a part that drifts off to $-\infty$ as $y\to x$?
Fortunately, such pathology is ruled out by Lemma~\ref{lem:sectorial-usc}, which says that
a slight enlargement of a sector for one member of a \RSF\ is a sector for all sufficiently
close members.

\begin{figure}
  \centering
\includegraphics[width=40mm]{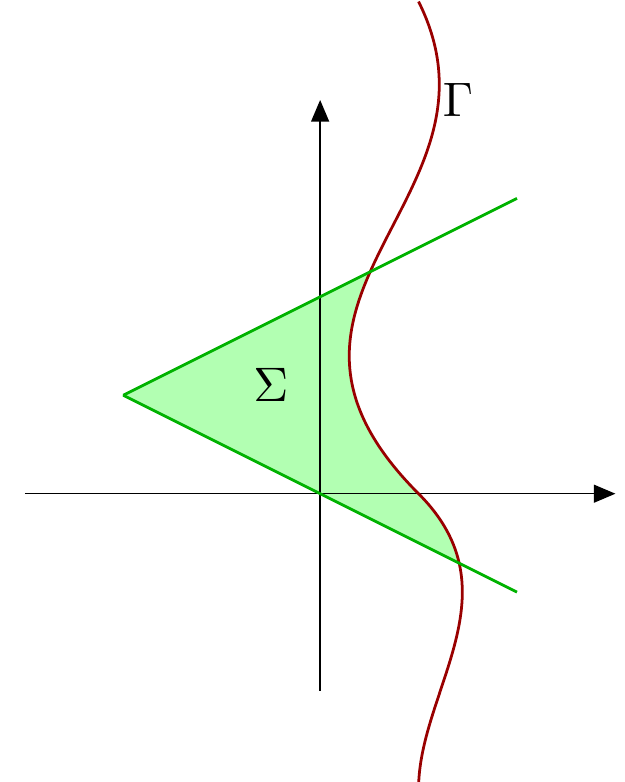}  
\caption{The concept of right-boundary.
  $\Sigma$ is a sector for $\Num A$, and the cross-hatched regions represent $\spec A$.
  The green region, containing all of $\spec A$ to the left of $\Gamma$,
  is surrounded by a contour bordered by parts of $\Gamma$ and the edges
  of $\Sigma$, and the precise choice of $\Sigma$ is irrelevant for Riesz-Dunford-Taylor
  integrals as in (\ref{eq:f(A|G)}). 
}
  \label{fig:right-boundary}
\end{figure}

More generally than a vertical line, we may start with a 
continuous curve $\Gamma \subset \Cmplx$ such that each horizontal line
\hbox{$\mathrm{Im}\ z = $ constant} intersects $\Gamma$ in exactly one point.
In other words, $\Gamma$ goes from bottom to top of the plane without overhangs,
as illustrated in Fig.~\ref{fig:right-boundary}.
Such a curve, with upward orientation, will be called a {\it right-boundary}.
Suppose $\Gamma$ is a right-boundary contained in $\res H_x$, and
let $\Sigma$ be a sector for $\frm{h}_x$ with vertex to the left of $\Gamma$
(Fig.~\ref{fig:right-boundary}).
Then we may form a closed contour by running along the lower edge of
$\Sigma$ away from the vertex until meeting $\Gamma$, then running
upward along $\Gamma$ until meeting the upper edge of $\Sigma$, and
then back to the vertex.
This contour, called $\tilde{\Gamma}$, encircles all the numerical range of
$\frm{h}_x$ to the left of $\Gamma$, hence the part of $\spec H_x$ in that region.
And, therefore, according to the preceding paragraph, $\tilde{\Gamma}$ also
encloses all of $\spec H_y$ lying to the left of $\Gamma$, for $y$ in some neighborhood
of $x$. Now we extend the notation in
(\ref{eq:P(G)}), (\ref{eq:A(G)}), and (\ref{eq:f(A|G)})
(as long as $f$ is holomorphic on the region to the left of $\Gamma$),
writing for instance $f(H_y|\Gamma)$ for the integral taken around $\tilde{\Gamma}$.
The point is that it does not matter how $\Gamma$ is completed to a closed contour as long
as all the spectrum to the left of $\Gamma$ is enclosed.
Since that can always be done (assuming $\Gamma \subset \res H_x$), the notation is
justified.

\subsection{Schatten classes}
\label{sec:Schatten}

The preceding part of this Section showed how we get a variety of holomorphic
maps $\Arr{\calU}{f}{\Lin(\sH)}$. What if the image of $f$ happens to be
in some restricted class of operators which has its own Banach space structure,
for instance, the trace-class operators $\Lin^1(\sH)$, the Hilbert-Schmidt
operators $\Lin^2(\sH)$, or more generally a Schatten $p$-class $\Lin^p(\sH)$?
Nearly automatic holomorphy in these situations is shown in 
Prop.~\ref{prop:automatic-holo-Lp} below.
Only the trace-class $\Lin^1(\sH)$ is used in this Section,
but other Schatten classes $\Lin^p(\sH)$ will be put to work in Section \ref{sec:free-energy}.

First, we recall some basic facts about the Schatten
$p$-classes\cite{Schatten,Gohberg+Krein,Simon-trace-ideals} that we will use.
\begin{defn}
For $1 \le p < \infty$, $\Lin^p(\sH)$ is the set of compact operators $T$ such that
$|T|^p\in\Lin^1(\sH)$, where \hbox{$|T| = (T^*T)^{1/2}$}.
\end{defn}
\begin{prop}
  \label{prop:Schatten}
The classes $\Lin^p(\sH)$ have the following properties.  
\begin{enumerate}
\item 
Equipped with the norm
\hbox{$\|T\|_p = (\Tr |T|^p)^{1/p}$}, $\Lin^p(\sH)$ is a Banach space.
\item
$\Lin^p(\sH)$ is also a two-sided $*$-ideal: \hbox{$\|ACB\|_p \le \|A\| \|C\|_p \|B\|$}
and it contains $C^*$ whenever it contains $C$.
\item
$\Lin^1(\sH)$ is the dual space of the compact operators $\Lin_0(\sH)$
with the usual operator norm,
while for \hbox{$1 < p < \infty$}, $\Lin^p(\sH)$ realizes the dual of $\Lin^{q}(\sH)$,
where \hbox{$p^{-1} + q^{-1} = 1$}, via the pairing \hbox{$(S,T) \mapsto \Tr ST$}.
On the other hand, the finite-rank operators are dense in $\Lin_0(\sH)$
as well as $\Lin^p(\sH)$ for $1 < p < \infty$.
Thus, {\em every} $\Lin^p(\sH)$ ($1\le p < \infty$) is the dual space of
a Banach space in which the finite-rank operators are dense.
\end{enumerate}
\end{prop}
\begin{prop}
  \label{prop:automatic-holo-Lp}
  Given: \hbox{$\Arr{\calU}{f}{\Lin(\sH)}$} holomorphic.
  If $f$ is a locally bounded map into $\Lin^p(\sH)$ ($1\le p < \infty$),
  then $f$ is holomorphic into $\Lin^p(\sH)$.
\end{prop}
\begin{proof}
For $B$ finite-rank, $\Tr f(x)B$ is
a finite sum of terms of the form $\inpr{\phi_\alpha}{f(x) \psi_\alpha}$, each of
which is holomorphic by hypothesis.
Hence, the result follows from the remark about density of such operators in
the pre-dual which precedes the Proposition together with  
Prop.~\ref{prop:wk*-holo}.
\end{proof}

\subsection{Finite rank}
\label{sec:finite-rank}

Prop.~\ref{prop:automatic-holo-Lp} does not quite give holomorphy
due to the hypothesis of local boundedness.
However, if we specialize to Riesz-Dunford-Taylor integrals and ask
that $P_x(\Gamma)$ have finite rank,
holomorphy into $\Lin^1(\sH)$ follows without an explicit local boundedness
assumption.
The next two well-known Lemmas encapsulate the simple key observations.
\begin{lem}
  \label{lem:rank-stability}
  If $P$ and $Q$ are projections (not necessarily orthogonal),
\hbox{$\|P-Q\| < 1$} implies that $\rank P = \rank Q$.
\end{lem}
\begin{proof}
If $\rng Q \ni \phi \mapsto P\phi$ is injective, then
\hbox{$\rank P \ge \rank Q$}, which suffices by symmetry of the situation.
However, for $\phi\in\rng Q$,
\begin{equation}
  \nonumber
\|P\phi\| = \|Q\phi + (P-Q)\phi\| \ge \|\phi\| - \|P-Q\| \|\phi\| > 0. 
\end{equation}
\end{proof}
\begin{lem}
  \label{lem:cts-into-trace-class}
  A continuous function into $\Lin(\sH)$ with range in operators
  of rank $\le N < \infty$ is actually continuous into $\Lin^1(\sH)$.
\end{lem}
\begin{proof}
{$\|A - B\|_1 \le (\rank A + \rank B) \|A - B\|$}.
\end{proof}

\begin{prop}
\label{prop:holo-into-L1}
$\rank P_x(\Gamma) = N < \infty$ implies that $x$ has a neighborhood $\calW$ 
such that $\rank P_y(\Gamma) = N$ for every $y\in\calW$, and
\hbox{$y \mapsto f(H_y|\Gamma)\colon \calW \to\Lin^1(\sH)$}
is holomorphic.
\end{prop}
\begin{proof}
Lemma~\ref{lem:rank-stability} ensures existence of $\calW$ such that
$\rank P_y(\Gamma) = N$ for $y\in\calW$.
Therefore $f(H_y|\Gamma)$ also has rank $N$ since it
maps $\rng P_y(\Gamma)$ into itself while annihilating $\ker P_y(\Gamma)$.
Lemma~\ref{lem:cts-into-trace-class} then completes the proof.
\end{proof}

\subsection{Eigenstate perturbation}
\label{sec:rank-1}

The extreme case is $\rank P_x(\Gamma)=1$.
Then we are in the venerable context of eigenstate perturbation.
A general rank-1 projection can be written as
\begin{equation}
\outpr{\phi}{\eta},\;\text{with}\; \inpr{\eta}{\phi} = 1\;\text{and } \|\phi\|=1,
\end{equation}
where $\phi$ and $\eta$ are determined up to a {common} phase factor $e^{i\theta}$.
Suppose, now, that ${\frm{h}}$ is a \RSF\ 
that $H_x$ has an isolated nondegenerate eigenvalue at $E_x$, and let
$\Gamma$ be a contour which separates $E_x$ from the rest of $\spec H_x$.
Then,
Prop.~\ref{prop:holo-into-L1} shows that
as $y$ varies in some neighborhood of $x$,
\begin{equation}
  \label{eq:Py(Gamma)}
  P_y(\Gamma) = \outpr{\phi_y}{\eta_y}, \;
  \inpr{\eta_y}{\phi_y} = 1, \; \|\phi_y\|=1,
\end{equation}
and
\begin{equation}
H_y(\Gamma) = E_y \outpr{\phi_y}{\eta_y},  
\end{equation}
with $P_y(\Gamma)$ and $E_y = \Tr H_y(\Gamma)$ holomorphic.
{\it A fortiori}, $E_y$ moves continuously with $y$ as long as it remains
separated from the rest of $\spec H_y$ ---
the {\it isolation condition}, for short.
As $y$ moves along any continuous curve in $\sX$ beginning at $x$ and respecting
the isolation condition $E_y$ can be continuously tracked, but if the path returns
to $x$, we may not return to $E_x$ unless the path can be contracted to a point
without violating the isolation condition.
Therefore, we consider $\calW$, a maximal {\em simply connected} open set
containing $x$ and with the isolation condition satisfied everywhere in $\calW$.
For $y$ in $\calW$, we can simply write $P_y$ and $E_y$, since the particular
choice of $\Gamma$ is immaterial.

Now, $E_y$ is holomorphic as a $\Cmplx$-valued function and
$P_y$ and $H_y$ as $\Lin^1(\sH)$-valued functions, for $y\in\calW$.
Therefore, for any bounded observable $B\in\Lin(\sH)$, its ``expectation''
\begin{equation}
y \mapsto  \Tr B \outpr{\phi_y}{\eta_y} = \inpr{\eta_y}{B\phi_y}  
\end{equation}
is holomorphic on $\calW$. The quotation marks are because this 
coincides with the usual notion of expectation only when $\eta_y=\phi_y$,
e.g., when $H_y$ is self-adjoint.

There are other interesting holomorphic quantities which do not fall into this category,
however. $E_y$ itself,
\begin{equation}
E_y =  \inpr{\eta_y}{H_y\phi_y},
\end{equation}
is one such.
The charge and current density are others when
our parameter space includes scalar and vector potentials.
This is because these quantities are the derivatives of $E_y$ with respect
to scalar and vector potential, respectively.
At a heuristic level, this claim is straightforward, but 
there are delicate details, which we will now check.

\begin{lem}[Hellmann-Feynman]
  \label{lem:Hellmann-Feynman}
  Suppose finite-rank projections $P_y$ and bounded operators $A_y$ depend
  differentiably on parameter $y$, and that $[P_y,A_y]=0$.
  Then $D_y \Tr P_y A_y = \Tr P_y D_yA_y$.
\end{lem}
\begin{proof}
($y$ subscripts will be suppressed for notational simplicity)
  Differentiating $P(1-P)=0$, deduce that $DP$ maps $\rng P$ into $\rng (1-P)$ and
  vice versa. Since both $\rng P$ and $\rng (1-P)$ are invariant under $A$, it
  immediately follows that $\Tr (DP)A = 0$ (put $(P + 1 - P)$ on each side and
  use cyclicity of trace).
\end{proof}

Since $H_y(\Gamma)$ is analytic, the preceding Lemma gives 
\begin{align}
  \label{eq:DE-integral}
  D_y E_y|_{y=x}
  &= \inpr{\eta_x}{DH_y(\Gamma)|_x\, \phi_x}
  \nonumber \\
  &= \oint_\Gamma
\inpr{\eta}{D_y\Rmap(\zeta,H_y)\phi} \zeta\frac{d\zeta}{2\pi i}.
\end{align}
($x$ subscripts are being omitted now, for simplicity.)

Now, may we may write
$D\Rmap(\zeta,H_y) {=} \Rmap(\zeta,H_y)D_yH_y\Rmap(\zeta,H_y)$?
A priori, this makes no sense since $H_y$ here is the full ({\em not} projected) operator.
However, if we understand $\Rmap(\zeta,H_y)$ as being in $\Lin(\sHp;\sHm)$
(see Sections \ref{sec:R-map} and \ref{sec:series}), so that
$\ilinpr{\eta}{\Rmap(\zeta,H_y)\phi} {=} \ilinpr{\eta}{(\hat{H}_y-\zeta)^{-1} \phi}$,
all is well.
Here, $\phi$ is considered as an element of $\sHp$, and $\eta$ of $\sHm$.
Then,
  \begin{equation}
D_y\inpr{\eta}{\Rmap(\zeta,H_y)\phi} 
 = \ilinpr{\eta}{(\hat{H}_y-\zeta)^{-1} D_y\hat{H}_y(\hat{H}_y - \zeta)^{-1} \phi}.
\end{equation}
To continue, we need 
\begin{lem}
    \label{lem:eta-eigenvector}
$H_y^*\eta_y = \overline{E}_y\eta_y$ and $\eta_y\in\sHp$.  
\end{lem}
\begin{proof}
First, note that $P_y^*\eta_y = \eta_y$.
Now (Prop.~\ref{prop:Riesz-Dunford-Taylor}),
$H_y = P_y H_y + (1-P_y)H_y(1-P_y)$, and $P_y$ commutes with $H_y$ on $\dom H_y$.
Therefore, 
\begin{align}
  \psi\in\dom H_y & \;\Rightarrow\;
                    \nonumber \\
&  \inpr{\eta_y}{H_y\psi}    
= \inpr{\eta_y}{P_y H_y \psi}    
 = \inpr{\eta_y}{H_y P_y \psi}    
                    \nonumber \\
= & E_y \inpr{\eta_y}{\phi_y}  \inpr{\eta_y}{\psi}
= E_y     \inpr{\eta_y}{\psi}.
\end{align}
This shows that $H_y^*\eta = \overline{E_y} \eta$.
Also, $\eta_y\in\sHp$,
because (Lemma~\ref{lem:adjoint-ops})
\hbox{$H_y^* = \Opp{0}{0}{\frm{h}_h^*}$} and $\frm{h}_y^*\in\sct{\calC}$
even if not in our parameterization.
\end{proof}

Using this Lemma, the previous display is rewritten as
\hbox{$-(\zeta - E_y)^{-2} \ilinpr{\eta}{D\hat{H}_y \phi}$}, which,
inserted into the contour integral (\ref{eq:DE-integral})
allows an easy evaluation. In conclusion,
\begin{prop}
\begin{equation}
  D_y E_y\Big|_{x} = D_y \Dbraket{\eta_x}{\frm{h}_y}{\phi_x}\Big|_x.
\end{equation}
\end{prop}
To do much with this requires explicit knowledge of $\frm{h}$. 

\subsubsection{charge/current density}
\label{sec:eigenstate-cc-density}

For a concrete case, 
consider a \RSF\ of Schr\"odinger forms as in Section \ref{sec:QM}.
The differentials of $E$ with respect to $u$ and ${\bm A}$ are linear forms
on a perturbation $\delta\!{u}$ or $\delta\!{\bm A}$
(the `$\delta$' doesn't actually have any independent meaning from our perspective),
given by
\begin{align}
D_{u} E \cdot\delta\!{u}
&  = \sum_\alpha \inpr{\eta}{\delta\!{u}(x_\alpha) \phi}
                           \nonumber \\
  \label{eq:rho}
&  \eqdef \int \delta\!{u}\, {\rho}\, d\underline{x},
\end{align}
and
\begin{align}
  D_{\bm A} E \cdot\delta\!{\bm A}
  &  =
 \inpr{ \cc{\delta\!{\bm A}}\, \eta }{ (i{\nabla} + {\bm A})\phi }
    + \inpr{(i\nabla+\cc{\bm A}) \eta }{ \delta\!{\bm A}\, \phi }
    \nonumber \\
  \label{eq:J}
&  \eqdef -\int \delta\!{\bm A}\cdot {\bm J}\, d\underline{x}
\end{align}
using the abbreviated notation of (\ref{eq:complex-A}).
These define the charge density $\rho$ and current density ${\bm J}$ of
the state in question.
In classical notation, one writes
\hbox{$\rho = {\delta E}/{\delta u}$} and
\hbox{${\bm J} = -{\delta E}/{\delta {\bm A}}$}.
More explicitly,
\begin{equation}
  \label{eq:rho-formula}
  \rho(x) = \sum_\alpha \int
  \cc{\eta}\phi|_{(x_{\alpha}=x)}  
  \, d\underline{x}_{-\alpha}
\end{equation}
and
\begin{equation}
  \label{eq:J-formula}
  {\bm J}(x) = 2 {\bm A}(x)\rho(x)
+ \sum_\alpha \int
  i(\cc{\eta}\overleftrightarrow{\nabla}\phi)|_{(x_{\alpha}=x)}  
  \, d\underline{x}_{-\alpha},
\end{equation}
where the notation means that integration is over all positions except those
of particle $\alpha$, which is set equal to $x$.

Of course, when $\frm{h}$ is not hermitian,
the physical interpretation of these as charge/current densities is rather unclear,
but the identifications are natural generalizations, indeed analytic continuations.

Restricted to hermitian $\frm{h}$, $\rho$ and 
$\bm J$ are $\Real$-analytic, but as maps into what Banach spaces?
Simplifying very slightly what we had in Section~\ref{sec:QM},
we take $u$ and ${\bm A}$ in
\hbox{$\sX_u = L^{3/2}(\Real^3)+L^\infty(\Real^3)$} and
\hbox{$\sX_{\bm A} = \vec{L}^{3}(\Real^3)+\vec{L}^\infty(\Real^3)$}, respectively.
As differentials of a scalar function on $\sX_u\times\sX_{\bm A}$, then,
$(\rho,{\bm J})$ is in $\sX_u^*\times\sX_{\bm A}^*$, a priori.
This is highly inconvenient due to the presence of the $L^\infty$ summands.
Fortunately, we can show that
\hbox{$\rho\in \sY_\rho = L^{3}\cap L^{1}$} and
\hbox{${\bm J}\in \sY_{\bm J} = \vec{L}^{3/2}\cap \vec{L}^1$}.
It then follows that $(u,{\bm A}) \mapsto (\rho,{\bm J})$ is analytic into
$\sY_\rho \times \sY_{\bm J}$ because\cite{Liu+Wang-68,Liu+Wang-69}
$\sX_u = \sY_\rho^*$, which implies that $\sY_\rho$ is embedded into
$\sX_{u}^*$ [$ = \sY_\rho^{**}$]
as a closed subspace, and similarly $\sY_{\bm J}$ into $\sX_{\bm A}^*$.
Here, we understand $L^p\cap L^q$ to be equipped with the max norm
\hbox{$\|f\| = \max(\|f\|_p,\|f\|_q)$}.

It suffices to show that $\rho$ and ${\bm J}$ are integrable, since
the integral forms (\ref{eq:rho},\ref{eq:J}), and the fact that they
induce linear functionals on $L^{3/2}$ and $L^3$, respectively, then
shows that $\rho\in L^3$ and ${\bm J}\in \vec{L}^{3/2}$.
Here are the required bounds:
First, from (\ref{eq:rho-formula}),
$\|\rho\|_1 \le N \|\eta\|^2 = N \|P^* P\| = N \|P\|^2$,
$P$ being the state projector [see (\ref{eq:Py(Gamma)})].
Then, from (\ref{eq:J-formula}), what was just shown establishes that
$\rho{\bm A}$ is integrable, and the Cauchy-Bunyakovsky-Schwarz inequality
shows that the second term is also, since $\eta,\phi\in\sHp$.
As discussed in the Introduction, these conclusions are relevant to
density functional theory (DFT),
current-density functional theory (CDFT), and magnetic-field density functional theory.

\section{Semigroups and statistical operators}
\label{sec:free-energy}

Whereas the ideas of the previous section trace their lineage back to the
primitive notion of inversion, the progenitor of this section is exponentiation.
We will study the operator family $e^{-\beta H}$ as $\beta$ ranges over
a vertex-zero sector and $H$ over operators associated with a \RSF.
In quantum statistical mechanics, $e^{-\beta H}$, assuming it is trace-class,
is the unnormalized statistical operator of a system with Hamiltonian $H$ at temperature
$T=\beta^{-1}$.
The trace, $Z_{\beta,H} = \Tr e^{-\beta H}$, is the partition function, and
$F_{\beta,H} = -\beta^{-1} \ln Z_{\beta,H}$ is interpreted as thermodynamic
free energy.
At nonzero temperature, the statistical operator and free energy play
roles analogous to those played by the ground state and ground state
energy at zero temperature.
Temperature, however, is not the only thermodynamic control parameter.
For a system with variable particle number(s), for instance, there are
chemical potentials $\mu_i$ for the various species, $i$. $\beta H$ should be replaced by
$\beta \left( H - \sum \mu_i N_i \right)$, where $N_i$ is the number of particles of
species $i$. This can be treated as a Hamiltonian on a Fock space with variable particle
number. Another thermodynamic parameter, volume can be incorporated in the form of a
confining potential.
In this way, we naturally move in the direction of considering the Hamiltonian
as being a highly variable object and studying the dependence of the statistical
operator and free energy on it.

This statistical interpretation ceases to be viable if the trace-class requirement
is dropped, but this more relaxed setting also has physical interest, especially
in connection with ideas around ``imaginary time'' evolution.
Here, the {\em semigroup} aspects come to the fore.
\hbox{$[0,\infty) \ni \beta \mapsto T(\beta) \equiv e^{-\beta H}$} should be the operator
semigroup generated by $-H$.
As Cor.~\ref{cor:Rmap-holo} showed that
the $\Rmap$-map $(\zeta,x) \mapsto \Rmap(\zeta,x) = (\zeta-H_x)^{-1}$
is holomorphic on its natural domain in $\Cmplx\times \calU$,
Cor.~\ref{cor:exp-holo} shows that the $\Emap$-map
$(\beta,x) \mapsto \Emap(\beta,x) = e^{-\beta H_x}$ is holomorphic,
where $\beta$ in the right half-plane $\CmplxRt$ is restricted only by the requirement
of sectoriality.
Section \ref{sec:free-energy-perturbation} considers a case where the statistical
interpretation is viable. With $H_0$ a lower-bounded self-adjoint operator
with resolvent in some Schatten class, and an \RSF\ in $\OF{H_0}$, $F_{\beta,x}$ is
holomorphic for $\beta$ in some neighborhood of $\Real_+$ and $x$ in some neighborhood
of zero. Similarly to the case of nondegenerate eigenstates considered in section
\ref{sec:rank-1}, this implies analyticity of (generalized) observables.
Charge-density and current-density are again examined in detail.

\subsection{Operator semigroups}
\label{eq:semigroups}

We begin with a recollection of some relevant
definitions\cite{Engel+Nagel,Engel+Nagel-big,Goldstein-semigps,Kato}.
A map $\Arr{[0,\infty)}{U}{\Lin(\sX)}$ is a {\em strongly continuous operator semigroup} if
\newline \noindent \textnormal{(1)} It respects the semigroup structure
of $[0,\infty)$: \hbox{$U(0)=\Id$} and \hbox{$U(s+t) = U(s)U(t)$}.
\newline \noindent \textnormal{(2)}
For each $x\in\sX$, the orbit map $t \mapsto U(t)x$ is continuous.

The {\em generator} $A$ of the semigroup is defined by
\begin{equation}
  A x = \lim_{t\downarrow 0} \frac{Ax - x}{t},
\end{equation}
$\dom A$ being the subspace on which the limit exists.
$A$ is a closed operator with dense domain
and for $x\in\dom A$, $\frac{d}{dt} U(t)x = U(t) Ax$
(e.g., Engel \& Nagel\cite{Engel+Nagel}, Thm II.1.4 and Lemma II.1.1).
The semigroup $U(t)$ is often denoted $e^{tA}$, which can be understood
in a very straightforward (power series) sense when $A$ is bounded.
A strongly continuous semigroup is necessarily locally bounded in operator norm.

If we leave everything above the same, except to expand the domain
from $[0,\infty)$ to $\oSec{0}{\theta}\cup\{0\}$ (also a semigroup),
$U$ is a {\em holomorphic} semigroup.
That the appelation is deserved follows from denseness of $\dom A$ and local
boundedness, which implies that $U$ is strongly holomorphic, and therefore
[Lemma~\ref{lem:st-holo} (e)] holomorphic $\oSec{0}{\theta} \to \Lin(\sX)$.

Now, if $H$ were bounded, $e^{-\beta H}$ could be obtained with a Riesz-Dunford-Taylor
integral of the function $e^{-\beta\zeta}$ along a contour surrounding the entire
spectrum. If $H$ is sectorial, though, its spectrum is unbounded only toward the
right in $\Cmplx$, where $e^{-\beta\zeta}$ is rapidly decreasing, assuming 
$|\arg \beta|$ is not too large.
This suggests that a contour such as $\Gamma$ in Fig.~\ref{fig:free-energy-contour}
might work. That it does so is the content of the following theorem, for the
proof of which we refer to the secondary literature.
\begin{figure}
  \centering
\includegraphics[width=65mm]{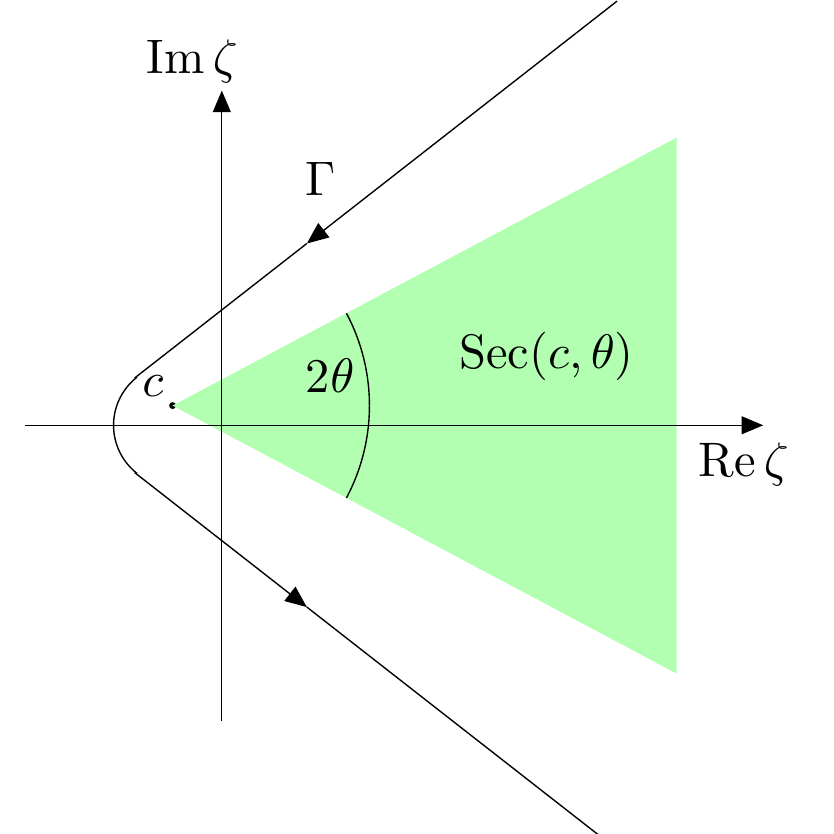}  
\caption{
The contour $\Gamma$ is adapted to the sector $\Sigma$.
The dashed line is the boundary of a dilation of $\Sigma$ and
$\Gamma$ lies exterior to it.
}
  \label{fig:free-energy-contour}
\end{figure}
\begin{defn}
  The contour $\Gamma$ in $\Cmplx$ parameterized by arc-length $s$
  is {\em adapted} to sector $\Sigma$ if
  $\re \Gamma(s) \to + \infty$ as $s \to \pm\infty$,
  and $\Gamma$ is exterior to some dilation
  of $\Sigma$  (item~\ref{item:sector}, Sec.~\ref{sec:sforms-1}).
\end{defn}
\begin{thm}
  \label{thm:holo-semigp}
  Let $A$ be a densely-defined operator with $\spec A$ contained in a sector $\Sigma$
  of half-angle $\theta$,
  such that
\begin{equation}
\zeta \not\in\Sigma' \;\Rightarrow\;
\|\Rmap(\zeta,A)\| \le \frac{M(\Sigma')}{|\zeta|+1}.
\end{equation}
for every dilation $\Sigma'$ of $\Sigma$.
Then, with $\Gamma$ a contour adapted to $\Sigma$,
a holomorphic semigroup
\hbox{$\oSec{0}{\tfrac{\pi}{2}-\theta} \to \Lin(\sH)$}
with generator $A$ is defined by
\begin{equation}
  \label{eq:holomorphic-semigroup-integral}
    \beta \mapsto e^{-\beta A}
= \int_\Gamma \Rmap(\zeta,A) e^{-\beta\zeta} \frac{d\zeta}{2\pi i}.    
\end{equation}
\end{thm}
\begin{proof}
See   
\S II.4 of Engel \& Nagel\cite{Engel+Nagel},
\S IX.1.6 of Kato\cite{Kato},
or \S X.8 of Reed \& Simon\cite{Reed+Simon}.
\end{proof}
Because $e^{-\beta A}$ is holomorphic into bounded operators, it
has a strong regularizing property not enjoyed by the generic operator semigroup:
\begin{cor}
  \label{cor:maps-into-dom-A}
  $\beta \mapsto e^{-\beta A}$ is a continuous linear map of $\sH$ into
  $\dom A$ (with the $A$-norm). 
\end{cor}

\subsection{The exponential map $\Emap$}
\label{sec:inverse-exponential}

Just as we earlier expanded the usual holomorphy of the resolvent $\Rmap(\zeta,H)$
with respect to the spectral parameter to find that it was holomorphic in
a parameterization of $H$ via a \RSF, we will in this subsection (Thm.~\ref{thm:exponential-holo})
expand the holomorphy of $\beta \mapsto e^{-\beta H}$ just discussed to include
$H$. If we imagine varying $A$ in (\ref{eq:holomorphic-semigroup-integral}),
we see that we should restrict to $A$ with spectrum in a sector to which
$\Gamma$ is adapted. Since we deal with operators coming from {\sqf}s,
we want to consider sectors for the numerical ranges, not the spectra.
\begin{notn}
  For a sector $\Sigma$, $\Op(\Sigma)$ denotes the set of closed, densely
  defined operators on $\sH$ with numerical range in $\Sigma$.
\end{notn}
A key ingredient of the theorem is the following lemma, which shows that
the resolvent bound in Thm.~\ref{thm:holo-semigp} is respected.
\begin{lem}
\label{lem:resolvent-bound}
Given sector $\Sigma$, and $\Sigma'$, a dilation of $\Sigma$,
there is a constant $M(\Sigma,\Sigma')$ 
such that
\begin{equation}
  \zeta\not\in\Sigma' \;\Rightarrow\;
  \|\Rmap(\zeta,H)\| < \frac{M(\Sigma,\Sigma')}{|\zeta|+1}.
\end{equation}
for every $H\in\Op(\Sigma)$.
\end{lem}
\begin{proof}
This is an immediate consequence of Prop.~\ref{prop:resolvent-outside-Num}.
\end{proof}
\begin{thm}
  \label{thm:exponential-holo}
  Let $\frm{h}$ be a \RSF.
  With the notation $H_y = \Opp{0}{0}{\frm{h}_y}$ as in Sec. \ref{sec:closure},
\begin{equation}
\Arr{\calU}{ y \mapsto e^{-H_y} }{\Lin(\sH)}
\end{equation}
is holomorphic.
\end{thm}
\begin{proof}
Let $\Sigma$ be an ample sector for $\frm{h}_x$.
Thus $\cl\Num H_x$ and, {\it a fortiori}, $\spec H_x$ is contained in $\Sigma$.
Furthermore, by Lemma~\ref{lem:sectorial-usc}, there is a neighborhood
$\calV$ of $x$ such that for $y \in \calV$, the same holds for $\spec H_y$.

Now, let $\Gamma$ be a contour adapted to $\Sigma$ (Fig.~\ref{fig:free-energy-contour})
parameterized by arc length $s$, and $\Gamma_n$ the restriction to $-n \le s \le n$, for $n\in\Nat$.
The integrals
\begin{equation}
  {\mathcal I}_n(y)
\defeq \int_{\Gamma_n} \Rmap(\zeta,{H}_y) e^{-\zeta} \frac{d\zeta}{2 \pi i},
\end{equation}
and ${\mathcal I}(y)$, the integral over the entire contour,
are well-defined on $\calV$.
Since $\Gamma_n$ is compact, Thm.~\ref{thm:resolvent-holo} guarantees that 
$y \mapsto {\mathcal I}_n(y)$ is holomorphic.

Finally, holomorphy of $\mathcal I$ will be secured by uniform convergence
\hbox{${\mathcal I}_n\to {\mathcal I}$} on $\calV$,
according to Prop. \ref{prop:convergent-sequences}.
Such convergence holds due to the damping factor $e^{- \re \zeta}$ in the
definition of ${\mathcal I}(y)$ combined with the resolvent bound in
Lemma~\ref{lem:resolvent-bound}, which holds uniformly on $\calV$.
\end{proof}
\begin{defn}
  For a \RSF\ $\frm{h}$, the \hbox{$\Emap$-map} is defined by
  \begin{equation}
\Emap(\beta,x) = e^{-\beta H_x}    
\end{equation}
on the domain
\begin{equation}
\Omega {\!\defeq\!\!\!}
\setof{(\beta,x)\in \Cmplx_{\text{rt}}\times\calU}{\beta H_x \text{ is sectorial}},    
\end{equation}
where $\Cmplx_{\text{rt}}$ is the open right half-plane $\re \beta > 0$.

As with the $\Rmap$-map, we may also write $\Emap(\beta,\frm{t})$ for a
particular \sqf\ $\frm{t}$, thinking of $\sct{\calC}$ as an \RSF\ parameterized over itself.
\end{defn}
\begin{cor}
  \label{cor:exp-holo}
Let $\frm{h}$ be a \RSF\ defined on $\calU$,
with associated family $x\mapsto H_x$ of closed operators.
Then $\Arr{\Omega}{\Emap}{\Lin(\sH)}$ is holomorphic.
\end{cor}

\subsection{Statistical operator and free energy}
\label{sec:free-energy-perturbation}

In quantum statistical mechanics, $e^{-\beta H_x}$ is used in the following way:
with real $\beta$ interpreted as inverse temperature,
the partition function is $Z_{\beta,x} = \Tr e^{-\beta H_x}$, the free energy is
$F_{\beta,x} = -\beta^{-1} \ln Z_{\beta,x}$, and the {\it statistical operator} is
$\rho_{\beta,x} = Z_{\beta,x}^{-1} e^{-\beta H_x}$. The latter describes the (mixed)
thermal state at inverse temperature $\beta$ under Hamiltonian $H_x$, so
that the thermal expectation of (bounded, at least) observable $B$ in this state is 
\begin{equation}
\langle B \rangle_x = \Tr \rho_{\beta,x} B.
\end{equation}
The basic condition for this to make mathematical sense is that $e^{-\beta H_x}$
be trace-class. When we generalize to allow non-real $\beta$ and non-self-adjoint
$H_x$, the additional condition that $Z_{\beta,x} \neq 0$ is required.

Phase transitions are generally identified with
points of non-analyticity of the free energy density in the thermodynamic limit
(quantity of matter tends to infinity at fixed temperature and pressure, or whatever
parameters are appropriate).
For simple lattice models in particular, it is easy to see that free energy density
is analytic for finite systems, while (not so easy to see) singularities can occur
in the thermodynamic limit. This is strongly connected with the dogma that phase transitions
are phenomena purely of the thermodynamic limit\cite{Kadanoff-09}.
One may well ask, however, to what extent we may rule out non-analyticity
with more realistic Hamiltonians and a greater, possibly infinite, number of parameters,
without any thermodynamic limit. This question is addressed here.
Thm.~\ref{thm:exponential-holo} is an important stepping stone,
but the conclusion of holomorphy into $\Lin(\sH)$ must be strengthened.

A very useful frame in which to think is that of an ``unperturbed'' Hamiltonian
with a polynomial-bounded energy density of states.
This seems to be about the right assumption, since it will allow a good perturbation
theory as we shall see, while being satisfied in the usual models.
For instance, for $N$ distinguishable particles moving in $d$ dimensions,
the density of states for a
harmonic oscillator hamiltonian is ${\mathcal O}(E^{3N-1})$, and for the
usual kinetic energy in a box with periodic boundary conditions, ${\mathcal O}(E^{(3N-2)/2})$.
Nontrivial quantum statistics or repulsive interactions only improve matters, by
the min-max principle.
The main theorem (\ref{thm:holo-statistical-op}) is framed in the context of
the space $\OF{H_0}$ of Sec.~\ref{sec:self-adjoint}, where $H_0$ is a lower-bounded
self-adjoint operator with resolvent in $\Lin^p(\sH)$ for some $p$, and says that
$e^{-\beta H_x}$ is holomorphic for $\beta$ in a nontrivial sector and $x$ in some
neighborhood of $0$. The key ideas involved are
Prop.~\ref{prop:automatic-holo-Lp}, bounding the integral
(\ref{eq:holomorphic-semigroup-integral}) simply by bounding the $\Rmap(\zeta,H_x)$,
and using elementary semigroup properties to get an $L^1(\sH)$ bound from an
$L^p(\sH)$ bound.

Here is the main result of this subsection.
\begin{thm}
  \label{thm:holo-statistical-op}
  If self-adjoint $H_0$ is such that $\Rmap(\zeta,H_0)$ is in $\Lin^p(\sH)$
  for one (hence every) resolvent point and some $1\le p < \infty$,
  then
\begin{equation}
\nonumber    
\Arr{\oSec{0}{\frac{\pi}{4}}\times B_{1}(\OF{H_0})}{\Emap}{\Lin^1(\sH)}
\end{equation}
is holomorphic.
\end{thm}
Proof of the theorem proceeds through four lemmas.
The first reduces the context from $\Lin^1(\sH)$ to $\Lin^p(\sH)$.
\begin{lem}
  \label{lem-Exp-from-Lp-to-L1}
Suppose 
\begin{equation}
\nonumber    
\Arr{\oSec{0}{\theta}\times\calU}{\Emap}{\Lin^p(\sH)}
\end{equation}
is locally bounded.
Then, ${\Emap}$ is holomorphic into $\Lin^1(\sH)$.
\end{lem}
\begin{proof}
  According to Thm.~\ref{thm:exponential-holo} and Prop.~\ref{prop:automatic-holo-Lp},
  what needs to be shown is
  that ${\Emap}$ is a locally bounded map into $\Lin^1(\sH)$. Given
  the hypotheses, though, that follows from the generalized H\"older inequality
  \begin{equation}
\|e^{-\beta H}\|_1 \le \|e^{-(\beta/p) H}\|_p^p .
  \end{equation}
\end{proof}
To make use of this Lemma, we need conditions which will ensure the hypothesized
local $\Lin^p$-boundedness.
In the next two Lemmas, sector $\Sigma$, and $\Sigma'$ a dilation of $\Sigma$,
and a point $\zeta_0 \not\in\Sigma'$ are understood as given, while
$H$ is arbitrary in $\Op(\Sigma)$.
They reduce the problem to one of bounding $\|\Rmap(\zeta_0,H)\|_p$.
\begin{lem}
  \label{lem:resolvent-bounded-uniformly-in-zeta}
\begin{equation}
  \| \Rmap(\zeta,H)\|_p \le C(\Sigma,\Sigma',\zeta_0)  \| \Rmap(\zeta_0,H) \|_p
\end{equation}
\end{lem}
\begin{proof}
Lemma~\ref{lem:resolvent-bound} ensures that the factor in square brackets
in the resolvent identity
\begin{equation}
\Rmap(\zeta,H) = [1+\Rmap(\zeta,H)(\zeta-\zeta_0)]\Rmap(\zeta_0,H),
\end{equation}
is bounded uniformly for $\zeta\not\in\Sigma'$.
\end{proof}
\begin{lem}
\label{lem:exponential-bounded-in-Lp}
\begin{equation}
\|e^{-\beta H}\|_p \le M(\Sigma,\Sigma',\zeta_0,\beta) \|\Rmap(\zeta_0,H)\|_p,
\end{equation}
with $M(\Sigma,\Sigma',\zeta_0,\beta)$ locally bounded in 
\hbox{$\beta\in\oSec{0}{\tfrac{\pi}{2} - \theta}$}, where $\theta$ is the
half-angle of $\Sigma'$.
\end{lem}
\begin{proof}
Let contour $\Gamma$ satisfy $\zeta_0 \in \Gamma \subset \Sigma'$
(hence, $\Gamma$ is adapted to $\Sigma$).
Then, 
\begin{align}
  \| e^{-\beta H} \|_p
  & = \Big\|    \int_\Gamma \Rmap(\zeta,H) e^{-\beta \zeta} \frac{d\zeta}{2\pi i}    \Big\|_p
    \nonumber \\
  & \le   \int_\Gamma \| \Rmap(\zeta,H) \|_p e^{- \re \beta \zeta} \frac{|d\zeta|}{2\pi}
    \nonumber \\
  & \le   C(\Gamma,\beta) \sup_{\zeta\in\Gamma} \| \Rmap(\zeta,H) \|_p 
    \nonumber \\
  & \le   M(\Sigma,\Sigma',\zeta_0,\beta) \| \Rmap(\zeta_0,H) \|_p.
\end{align}
The third line follows since
\hbox{$\int_\Gamma e^{- \re \beta \zeta} {|d\zeta|} < \infty$},
and the fourth line is by Lemma~\ref{lem:resolvent-bounded-uniformly-in-zeta}.
\end{proof}
\begin{lem}
  \label{lem:perturbed-resolvent-Lp-bound}
If $\dom H \subseteq \dom A$  and
\hbox{$\| A\Rmap(\zeta_0,H_0) \| < 1$}, then
\begin{equation}
    \nonumber
    \| \Rmap(\zeta_0,H+A) \|_p
    \le \| ( 1 + A\Rmap(\zeta_0,H))^{-1}\|
    \| \Rmap(\zeta_0,H) \|_p
\end{equation}
\end{lem}
\begin{proof}
Immediate.  
\end{proof}
\begin{proof}[Completion of Proof of Thm.~\ref{thm:holo-statistical-op}]
Now it is merely a matter of stringing the pieces together. 
Prop.~\ref{prop:automatic-holo-Lp} asserts that local boundedess of
$e^{-\beta H_x}$ in $\Lin^1(\sH)$ suffices to establish holomorphy,
Lemma~\ref{lem-Exp-from-Lp-to-L1} shows that $\Lin^1(\sH)$ can be replaced
by $\Lin^p(\sH)$;
Lemma~\ref{lem:exponential-bounded-in-Lp} that we only need
a local bound on $\Rmap(\zeta_0,H_x)$;
and Lemma~\ref{lem:perturbed-resolvent-Lp-bound}
shows how big the perturbation can be. According to the definition of
$\OF{H_0}$ [see Section~\ref{sec:self-adjoint}, especially Prop.~\ref{prop:X(H)}],
it suffices that $\|\frm{t}-\frm{h}\|_{H_0}< 1$.
The restriction on $\beta$ is needed to insure that $\beta H_x$ is sectorial
for all $x$ in $B_1(\OF{H_0})$.
\end{proof}
\subsection{Thermal expectations}
\label{sec:thermal-expectations}

This subsection is concerned with consequences of Thm.~\ref{thm:holo-statistical-op}.
In other words, what do we do with the holomorphic statistical operator?
We Suppose given a \RSF\ in $B_1(\OF{H_0})$, and adopt the notational convention
that $\frm{h}_x$ corresponds to the operator $H_0 + T_x$ on $\dom H_0$
(i.e., this is $\Opp{0}{0}{\frm{h}_x}$).
We can be fairly explicit about the Taylor series expansion of $e^{-\beta(H_0+T_x)}$.
By appeal to Cor.~\ref{cor:maps-into-dom-A}, 
\begin{equation}
  \label{eq:perturbed-semigroup}
e^{-\beta(H_0+T_x)} = \int_0^1 e^{-s\beta(H_0+T_x)} (-\beta T_x) e^{-(1-s)\beta H_0} \,ds
\end{equation}
for any $\frm{h}_x\in\OF{H_0}$. Iteration shows that the $n$-th term of the
Taylor series has the familiar form
\begin{equation}
  (-\beta)^n\int_{\substack{ s\ge 0 \\ \sum s_k = 1}}
  e^{-s_{n+1}\beta H} T_x e^{-s_n \beta H}\ldots T_x e^{-s_1 \beta H} \,d\underline{s}.
\end{equation}
Thm.~\ref{thm:holo-statistical-op} implies that this actually converges for small enough $x$.

In the following, we will be concerned only with the first term, however.
For $(\beta,x)$ in some neighborhood of $\Real_+ \times \{0\}$,
$Z_{\beta,x}$ is nonzero and therefore the free energy $F_{\beta,x}$ is
well-defined and holomorphic. According to (\ref{eq:perturbed-semigroup}),
the derivative (also holomorphic) 
$-\beta D_x F_{\beta,x}$ is the expectation value $\Tr \rho_{\beta,x} D_x T_x$.

\subsubsection{charge/current density}
\label{sec:thermal-cc-density}

Parallel to the treatment of properties of energetically-isolated eigenstates in
Section~\ref{sec:rank-1}, we will consider charge and current-density in the thermal
context for an system of $N$ nonrelativistic particles in a three-dimensional box,
under a periodic boundary condition Hamiltonian consisting of a kinetic energy
operator $K_{\bm A} = \sum_{\alpha=1}^{N} \left|i\nabla_\alpha + {\bm A}(x_\alpha)\right|^2$,
a one-body potential operator $U_u = \sum_\alpha u(x_\alpha)$, and 
a repulsive two-body interaction $V_v = \tfrac{1}{2}\sum_{\alpha\neq\beta} v(x_\alpha-x_\beta)$.
The variables here are $u$ and ${\bm A}$.
$H_0 = K_{\bm A} + V_v$ has a polynomial-bounded density of states, hence
Thm.~\ref{thm:holo-statistical-op} applies and the free energy is holomorphic.
Charge and current-density
are obtained by differentiating the free energy with respect to
$u$ and ${\bm A}$, respecively, hence are also holomorphic
if the perturbed Hamiltonians comprise a \RSF\ in $\OF{H_0}$.

Now, apply the result of Kato (\S~5.5.3 of Kato\cite{Kato},
Thm.~6.2.2 of de~Oliveira\cite{deOliveira} or
Example 13.4 of Hislop \& Sigal\cite{Hislop+Sigal})
that a potential in $L^2(\Real^3) + L^\infty(\Real^3)$
is relatively bounded with respect to the ${\bm A}=0$ kinetic energy operator $-\Delta$
with relative bound zero (See Def.~\ref{def:relative-operator-bound}).
Since the system is confined to a box, a bounded potential is automatically square-integrable.

For the kinetic energy operator
\begin{equation}
|i\nabla + {\bm A}|^2 = -\Delta + 2i{\bm A}\cdot\nabla + i\,{\mathrm{div}}{\bm A} + |{\bm A}|^2,
\end{equation}
${\bm A}$ must be restricted so that each of the last three terms is adequately tame.
It will suffice that ${\bm A}\in \vec{L}^4(\text{Box})$ if we work in Coulomb gauge,
i.e., $\mathrm{div}{\bm A}=0$, or in Fourier components, ${\bm q}\cdot\tilde{\bm A}({\bm q})=0$.
We will denote this subspace of ``transverse'' vector fields by
$\vec{L}^4(\text{Box})_{\text{trans}}$.
That restriction obviously takes care of the divergence term.
\begin{equation}
\| |{\bm A}|^2 \|_{L^2} \le c \| {\bm A} \|_{L^4}^2
\end{equation}
by H\"older's inequality. 
Finally, since the box is bounded, $\vec{L}^3(\text{Box})$ is continously embedded in
$\vec{L}^4(\text{Box})$, so (\ref{eq:Holder-Sobolev}) demonstrates a suitable bound for
${\bm A}\cdot \nabla\psi$ when \hbox{$\psi \in \dom (-\Delta)$}.
For the scalar potential, no trickery is required to
apply the result cited above. Simply assume $u,v\in L^2(\text{Box})$.

Thus, we obtain an \RSF\ in $\OF{H_0}$ defined for $x \equiv (u,{\bm A})$ on
some neighborhood of the origin in
$L^2(\text{Box}) \times\vec{L}^4(\text{Box})_{\text{trans}}$.
The charge/current density $(\rho,{\bm J}) = -\beta D_xF_{\beta,x}$
is then an analytic function of $x$ valued in
$L^2(\text{Box}) \times \vec{L}^{4/3}(\text{Box})$.

\section{Summary}
\label{sec:summary}

Here is a summary of the apparatus developed here, from an application-oriented
perspective.
The starting point is a family \hbox{$\Arr{\sX\supseteq \calU}{\frm{h}}{\sct{\SF}(\sK)}$}
of closable, mutually relatively bounded, sectorial {\sqf}s parameterized over $\calU$.
Thinking of these as generalized Hamiltonians, sectoriality is an appropriate generalization
of lower-bounded and hermitian, which allows use of holomorphy. If quantities related to these
forms $\frm{h}_x$ or their associated operators $H_x$ are holomorphic in the parameter
$x\in\calU$, then real analyticity results for proper Hamiltonians by restriction.
If $x\mapsto\frm{h}$ is a \RSF, then 
holomorphy of $(\zeta,x) \mapsto \Rmap(\zeta,x) = (\zeta-H_x)^{-1}$ and
$(\beta,x) \mapsto \Emap(\beta,x) = e^{-\beta H_x}$, as maps into $\Lin(\sH)$,
is secured on natural domains.
This is the content of Cors.~\ref{cor:exp-holo} and \ref{cor:Rmap-holo}, respectively.
Prop.~\ref{prop:regular-sectorial} provides a few sets of convenient criteria for
$x\mapsto\frm{h}$ to be a \RSF.
One of these is,
(a) G-holomorphy:
for each $x\in\calU$, $w\in \sX$, and $\psi\in\sK$,
\hbox{$\zeta \mapsto \frm{h}_{x+\zeta w}[\psi]$} is holomorphic on
some neighborhood of the origin in $\Cmplx$; and
(b) local boundedness:
each $x\in\calU$ has a neighborhood such that $\frm{h}_y$ is bounded
uniformly relative to $\frm{h}_x$ for $y$ in that neighborhood.
The practicality of these criteria is demonstrated in Section~\ref{sec:QM},
where an \RSF\ of multi-particle Schr\"odinger forms is constructed.

The $\Rmap$-map and $\Emap$-map are themselves mostly means to an end.
An important tool in using them is Prop.~\ref{prop:automatic-holo-Lp},
which says that either is actually holomorphic into
the Schatten class $\Lin^p(\sH)$ (not just into $\Lin(\sH)$) if it
is merely locally bounded into $\Lin^p(\sH)$.
Using this, we can effectively deal with properties of isolated eigenstates,
or of thermal states when the $\Emap$-map is verified to be locally bounded
into trace-class operators (Thm.~\ref{thm:holo-statistical-op}).
Particularly interesting are derivatives of the
energy or free energy with respect to scalar potential $u$ or vector potential ${\bm A}$,
which give (expectation of) charge-density and current-density, respectively.
As differentials of holomorphic functions, these are automatically holomorphic themselves.
In the case of isolated eigenstates (Section \ref{sec:eigenstate-cc-density}),
$(\rho,{\bm J})$ is analytic in
\hbox{$(L^{3}(\Real^3) \cap L^1(\Real^3))\times(\vec{L}^{3/2}(\Real^3)\cap \vec{L}^1(\Real^3))$}.
as function of $(u,{\bm A})$ in 
$(L^{3/2}+L^\infty) \times (\vec{L}^3+\vec{L}^\infty)$.
For thermal states, additional restrictions are required on the potentials to
ensure existence of the free energy. For a system in a box
(Section~\ref{sec:thermal-cc-density}),
$(\rho,{\bm J})$ is analytic in $L^{2}(\text{Box}) \times \vec{L}^{4/3}(\text{Box})$
as function of $(u,{\bm A})$ in $L^{2} \times \vec{L}^4$,
with ${\bm A}$ in Coulomb gauge.

%


\end{document}

%% file: macros.tex
\newcommand{\defeq}{ {\kern 0.2em}:{\kern -0.5em}={\kern 0.2 em} }  
\newcommand{\eqdef}{ {\kern 0.2em}={\kern -0.5em}:{\kern 0.2 em} }  
\newcommand{\setof}[2]{\left\{ #1 \;\middle|\; #2 \right\}}  

\newcommand{\cc}[1]{\overline{#1}}  

\newcommand{\Arr}[3]{{#2}\colon {#1} \rightarrow {#3}} 
\newcommand{\Arrtop}[3]{{#1} \xrightarrow{{#2}} {#3}} 
\newcommand{\Type}[2]{\colon {#1} \rightarrow {#2}}    

\newcommand{\Id}{\mathrm{id}}  
\DeclareMathOperator{\rng}{rng}  

\DeclareMathOperator{\cl}{\mathrm{cl}}

\newcommand{\Nat}{{\mathbb N}}     
\newcommand{\Real}{{\mathbb R}}    
\newcommand{\Cmplx}{{\mathbb C}}   

\newcommand{\dom}{\mathrm{dom}\,}

\newcommand{\subdiff}
{\underline{D}}
\newcommand{\supdiff}
{\overline{D}}
\newcommand{\pair}[2]{\left\langle {#1} \,,\, {#2} \right\rangle}
\DeclareMathOperator{\re}{\mathrm{Re}}

\DeclareMathOperator{\res}{\mathrm{res}}
\DeclareMathOperator{\spec}{\mathrm{spec}}
\newcommand{\ilinpr}[2]{\langle {#1} | {#2} \rangle}    
\newcommand{\inpr}[2]{\left\langle {#1} \middle| {#2} \right\rangle}    
\newcommand{\outpr}[2]{\left| {#1} \middle\rangle\middle\langle {#2} \right|}   
\newcommand{\Dbraket}[3]{\left\langle {#1} \middle| {#2} \middle| {#3} \right\rangle}
\DeclareMathOperator{\Tr}{\mathrm{Tr}}

\newcommand{\slct}[1]
{\left[ {#1} \right] }

\newcommand{\V}
{B'}
\newcommand{\VC}
{{B_\Cmplx}'}

\newcommand{\frm}[1]{\mathsf{#1}}
\newcommand{\Lin}{{\mathcal L}}
\newcommand{\Linv}{\Lin_{\mathrm{iso}}}
\newcommand{\Lincl}{\Lin_{\mathrm{cl}}}
\newcommand{\sK}{\mathscr K}
\newcommand{\sKc}{\widetilde{\sK}}
\newcommand{\sH}{\mathscr H}
\newcommand{\sHz}{\sH}
\newcommand{\sHp}{\sH_{+}}
\newcommand{\sHm}{\sH_{-}}
\newcommand{\sX}{\mathscr X}
\newcommand{\sY}{\mathscr Y}
\newcommand{\sZ}{\mathscr Z}

\newcommand{\calC}{\mathcal C}
\newcommand{\calW}{\mathcal W}
\newcommand{\calU}{{\mathcal U}}
\newcommand{\calV}{{\mathcal V}}

\newcommand{\Op}{\mathrm{Op}}
\newcommand{\Opp}[3]{{ [ {#3} ]_{\scriptscriptstyle{#1}}^{ \scriptscriptstyle{#2}} }}

\DeclareMathOperator{\Num}{\mathrm{num }}
\DeclareMathOperator{\rank}{\mathrm{rank }}
\newcommand{\SF}{\mathsf{SF}}
\newcommand{\OF}[1]{{\mathcal B}({#1})}
\newcommand{\sqf}{s-form}
\newcommand{\RSF}{regular sectorial family}
\newcommand{\ktiny}{\prec{\kern -0.8em}\prec}

\newcommand{\cSec}[2]{\overline{\mathrm{Sctr}}\left({#1},{#2}\right)}
\newcommand{\oSec}[2]{\mathrm{Sctr}\left({#1},{#2}\right)}

\newcommand{\dist}{\mathrm{dist}}
\newcommand{\sct}[1]{{#1}^\triangleleft}
\newcommand{\relbd}{\preccurlyeq}
\newcommand{\subsub}[1]{_{_{#1}}}
\newcommand{\kfrm}[1]{\frm{k}_{#1}}
\newcommand{\ufrm}[1]{\frm{u}_{#1}}
\newcommand{\vfrm}[1]{\frm{v}_{#1}}

\newcommand{\Rmap}{{\mathcal R}}
\newcommand{\Emap}{{\mathcal E}}
\newcommand{\CmplxRt}{\Cmplx_+}